\documentclass[11pt,reqno]{amsart}

\usepackage{amsmath,amssymb,amsfonts,latexsym,amscd,epsfig,psfrag}

\theoremstyle{plain}
\newtheorem{theorem}{Theorem}[section]

\newtheorem{proposition}[theorem]{Proposition}

\newtheorem{lemma}[theorem]{Lemma}
\newtheorem{corollary}[theorem]{Corollary}
\newtheorem{remark}[theorem]{Remark}

\def\eqnn{\begin{eqnarray*}}
\def\eeqnn{\end{eqnarray*}}
\def\eqn{\begin{eqnarray}}
\def\eeqn{\end{eqnarray}}

\def\bal{\begin{align}}
\def\eal{\end{align}}

\def\prf{\begin{proof}}
\def\endprf{\end{proof}}

\def\R{{\mathbb R}}

\newcommand{\nc}{\newcommand}

\nc{\lam}{\lambda}
\nc{\G}{\Gamma}
\nc{\g}{\gamma}
\nc{\al}{\alpha}
\nc{\del}{\delta}
\nc{\om}{\omega}
\nc{\Om}{\Omega}
\nc{\Omt}{\tilde{\Omega}}
\nc{\ta}{\tau}
\nc{\w}{\omega}
\nc{\io}{\iota}
\nc{\h}{\theta}
\nc{\z}{\zeta}
\nc{\s}{\sigma}
\nc{\Si}{\Sigma}
\nc{\Lam}{\Lambda}

\def\cH{{\mathcal H}}

\def\cM{{\mathcal M}}

\def\cS{{\mathcal S}}

\def\e{\epsilon}

\def\g{\sqrt\alpha}

\def\supp{{\rm supp}}

\def\n{\nabla}

\def\Fo{\mathfrak{F}}

\def\1{{\bf 1}}
\def\id{\chi}

\def\Aop{\psi}

\def\cSd{\cS_\mu}
\def\Dop{\phi}

\def\Jop{{\mathcal J}_\sigma}
\def\JopN{{\mathcal J}_0}

\def\A{A}
\def\Af{A_\sigma}

\def\Eeff{E_{{\rm eff},\sigma}}

\def\Eg{E_{\sigma}}

\def\He{H^{V}_{\sigma}}
\def\Heff{H_{{\rm eff},\sigma}}
\def\Hef{H_{\rm eff}}
\def\Hn{H_{\sigma}}
\def\HV{H^{V}}
\def\h{H}

\def\Kn{K_\sigma}
\def\Ve{V_\e}
\def\Pf{P_f}
\def\Ptot{P_{tot}}

\def\holder{\theta}
\def\Holdexp{\holder }

\def\nablE{\nabla E_\sigma}
\def\Phsig{\Phi_{\sigma}}
\def\PhiN{\Phi}
\def\Psig{\Psi_{\sigma}}

\def\spec{{\rm spec}}

\newcommand{\DETAILS}[1]{}

\numberwithin{equation}{section}

\begin{document}

\parskip=8pt

\title[Effective dynamics in non-relativistic QED]
{Effective dynamics of an electron coupled to an external potential
in non-relativistic QED}

\author[V. Bach]{Volker Bach}
\address[V. Bach]{Institut fuer Analysis und Algebra
Carl-Friedrich-Gauss-Fakultaet,
Technische Universitaet Braunschweig,
38106 Braunschweig,
Germany }
\author[T. Chen]{Thomas Chen}
\address[T. Chen]{Department of Mathematics, University of Texas at Austin, Austin TX 78712, USA}
\email{tc@math.utexas.edu}
\author[J. Faupin]{J\'er\'emy Faupin}
\address[J. Faupin]{Institut de Math{\'e}matiques de Bordeaux \\
UMR-CNRS 5251, Universit{\'e} de Bordeaux 1 \\
351 cours de la lib{\'e}ration, 33405 Talence Cedex, France}
\email{jeremy.faupin@math.u-bordeaux1.fr}
\author[J. Fr\"ohlich]{J\"urg Fr\"ohlich}
\address[J. Fr{\"o}hlich]{Institut f{\"u}r Theoretische Physik, ETH H{\"o}nggerberg, CH-8093 Z{\"u}rich, Switzerland}
\address{Present address: School of Mathematics, The Institute for
Advanced Study, Princeton, NJ 08540; visit supported by
'The Fund For Math' and 'The Monell Foundation'.}
\email{juerg@phys.ethz.ch}
\author[I.M. Sigal]{Israel Michael Sigal}
\address[I.M. Sigal]{Department of Mathematics, University of Toronto, Toronto, ON M5S 2E4, Canada}
\email{im.sigal@utoronto.ca}

%\date{\Datum}

\begin{abstract}
In the framework of non-relativistic QED,  we show that the renormalized mass of the electron (after having taken into account radiative corrections)  appears as the kinematic mass in its response to an external potential force. Specifically, we study the dynamics of an electron in a slowly varying external potential and
with slowly varying initial conditions
and prove that, for a long time, it is accurately described by an associated effective dynamics of a
Schr\"odinger electron in the same external potential
and  for the same initial data,  with a kinetic energy operator
determined by the renormalized dispersion law of the translation-invariant QED model.
\end{abstract}

\maketitle

\begin{center}
%{\color{Red}
{\em 
This paper is dedicated to the memory of Walter Hunziker - teacher and friend.}
%}
\end{center}

\section{Introduction}

In this paper we show that the renormalized mass of the electron, taking into account radiative corrections due to its interaction with the quantized electromagnetic field, and the kinematic mass appearing in its response to a slowly varying external potential force are identical. Our analysis is carried out within the standard framework of  non-relativistic quantum electrodynamics (QED).  The  renormalized electron mass, $m_{\rm ren}$, is defined as the inverse curvature at zero momentum of the energy (dispersion law), $E(p)$, of a dressed electron as a function of its momentum $p$ (no external potentials are present),  i.e.,  $m_{\rm ren}= E''(0)^{-1}$, while the kinematic mass of the electron enters the (effective) dynamical equations when it moves under the influence of an external potential force. 

Our starting point is the dynamics generated by the Hamiltonian, $\HV$, describing a non-relativistic  electron interacting with the quantized
electromagnetic field and moving under the influence of a slowly varying potential, $V_\e$.  We consider the time evolution of dressed one-electron states parametrized by
wave functions $u_0^\e\in H^1(\R^3)$, with $\|u_0^\e\|_{L^2}=1$ and 
$\|\nabla u_0^\e\|_{L^2}\le  \e^\kappa$, with $0\leq\kappa<\frac13$, 
and prove that their evolution is accurately
approximated, during a long interval of time, by an effective Schr\"odinger dynamics generated by  the one-particle Schr\"odinger operator
\eqn\label{Heff-def-1-0}
	\Hef \, := \, E(-i\nabla_x)  \, + \, \Ve(x) \,,
\eeqn with kinetic energy given by the dispersion law $E(p)$.
This result is in line with the general idea that any kind of physical dynamics is an effective dynamics that can ultimately be derived from a more fundamental theory. 
While results of a similar nature have been proven for quantum-mechanical particles interacting with $massive$ bosons, \cite{spte}, 
ours is the first result covering  electrons interacting with photons (or, more generally, 
{\it massless} bosons) and revealing effects of radiative corrections to the electron mass.
%ours is the first result covering the physically more interesting situation of electrons interacting 
%with $massless$ bosons (photons) and revealing effects of radiative corrections to the electron mass. 
Our derivation relies in an essential way
on recent regularity results on the mass shell, i.e., the
ground state energy and the corresponding ground state vector
as a function of total momentum \cite{cfp1,cfp2}. 
An interesting result on the effective dynamics of two heavy particles interacting via exchange of  massless bosons has previously been obtained in \cite{tt}.  

In the usual model of non-relativistic QED, the Hilbert space of states of a system consisting of a single electron and arbitrarily many photons (described in the Coulomb gauge) is given by
\eqn
	\cH \, := \, L^2(\R^3) \, \otimes \, \Fo \,,
\eeqn
where $L^2(\R^3)$ is the Hilbert space of square-integrable wave functions describing the electron degrees of freedom, (electron spin is neglected for notational convenience).
The space $\Fo$ is the Fock space of physical states of photons,
\begin{equation*}
	\Fo  :=  \bigoplus_{n\geq0} \Fo_n.
\end{equation*}
Here
$\Fo_n \, := \, {\rm Sym} ( \, L^2(\R^3 \, \times \, \{+,-\} \, ) \, )^{\otimes n}$
denotes the physical Hilbert space of states of $n$ photons.
The Hamiltonian acting on the space $\mathcal{H}$ is given by the expression
\eqn\label{He-def-0}
	\HV \, := \, \h \, + \, \Ve \otimes\1_f \, ,
\eeqn
where $H$ is the generator of the dynamics of a single, freely moving
non-relativistic electron minimally coupled to
the quantized electromagnetic  field, i.e.,
\eqn\label{Hn-def-0}
	\h  \, := \, \frac12( \,- i\nabla_x \otimes \1_f \, + \, \g \A(x) \, )^2
	\, + \, \1_{el} \otimes H_f \, ,
\eeqn
and where $\Ve(x):=V(\e x)$ is a slowly varying potential, with $\e>0$ small;
its precise properties are formulated in Theorem \ref{thm-main-1-0}  below.
Furthermore,
\eqn\label{Af-def-0-1}
	\A(x) \, := \,
	\sum_{\lambda}\int_{|k| \le 1}
	\, \frac{dk}{|k|^{1/2}} \,
	\{ \,  \e_\lambda(k) \,  e^{ikx} \otimes a_\lambda(k) \, + \, h.c. \, \}
\eeqn
denotes the quantized electromagnetic vector potential in the Coulomb gauge
with an ultraviolet cutoff imposed, $ |k|\leq 1$, and
\eqn\label{Hf-def}
	H_f \, := \, \sum_{\lambda} \int dk \, |k| \,  a_\lambda^*(k) \,  a_\lambda(k)
\eeqn
is the photon Hamiltonian.
In Eqs. \eqref{Af-def-0-1} and \eqref{Hf-def}, $a_\lambda^*(k)$, $a_\lambda(k)$ are
the usual photon creation- and annihilation operators, $\lambda=\pm$
indicates photon helicity, and $\e_\lambda(k)$ is a polarization vector perpendicular
to $k$ corresponding to helicity $\lambda$.
%{\color{Red}
We note that all results in this paper hold for sufficiently small values of the fine structure
constant, $0<\alpha\ll1$.
%}

We observe that the Hamiltonian $\h$ is translation-invariant, in the sense that $H$ commutes with translations, $T_{y}:  \Psi(x)\to  e^{ iy\cdot P_f} \Psi(x+y)$, for $y\in\R^3$,
where  $\Pf:=\sum_\lambda\int dk \, k \, a_\lambda^*(k)a_\lambda(k)$ is the momentum operator of the quantized radiation field. Hence $H$ commutes with the total momentum operator 
\eqn\label{Ptot}
	\Ptot \, := \, -i\nabla_x\otimes\1_f \, + \, \1_{el}\otimes \Pf,
\eeqn
of the electron and the photon field: $[\h, \Ptot]=0$. It follows that $H$ can be decomposed as a direct integral
\begin{equation}\label{H-fib-deco}
U\h U^{-1} = \int_{\mathbb{R}^3}^{\oplus} \h(p) dp,
\end{equation}
of fiber operators, $H(p)$, over the spectrum of $\Ptot$,  where $H(p)$ is defined on the fiber space $\mathcal{H}_{p}\cong\Fo$ in the direct integral decomposition, $\cH\cong\int^{\oplus}_{\mathbb{R}^3}dp \, \cH_p$, of $\mathcal{H}$. The operator $U : \cH \to \int^{\oplus}dp \, \cH_p$ is a generalized Fourier transform defined on smooth, rapidly decaying functions,
\begin{equation}\label{eq:fourier1} 
		(U \Psi)(p) :=( Fe^{i P_f\cdot x}\Psi)(p) = 
		(2\pi)^{-3/2}\int_{\mathbb{R}^3}e^{-i(p-P_f)\cdot x}\Psi(x) dx, 
\end{equation}
where $F$ is the standard Fourier transform for Hilbert space-valued functions, 
\begin{equation*}
(F\Psi)(p)= (2\pi)^{-3/2}\int_{\mathbb{R}^3}e^{-i p\cdot x}\Psi(x) dx. 
\end{equation*}
For  smooth, rapidly decaying vector-valued functions $\Phi (p) \in \cH$, its inverse is given by
 \begin{equation}\label{eq:fourier2} 
 (U^{-1}\Phi)(x) :=e^{-i P_f\cdot x}(F^{-1}\Phi)(x)
 =(2\pi)^{-3/2}\int_{\mathbb{R}^3} e^{ix\cdot(p-P_f)}\Phi (p) dp.
\end{equation}
We note that
\eqn \label{H-fib-deco'} 
	(U\h \Psi)(p) \, = \, \h(p)  (U\Psi)(p)
	\; \; \; , \; \; \; \; 	
	(U {\Ptot\psi})(p)  \, = \,p \,  
	(U\psi)(p) \,.
\eeqn
Since $U$ is the composition of two unitary operators, $ e^{i P_f\cdot x}$ and the standard Fourier transform $F$, it is unitary, too, and Eq. \eqref{eq:fourier2}  defines its inverse.

We define  creation- and annihilation operators, $b_{\lambda}^{*}(k)$ and $b_{\lambda}(k)$, on the fiber spaces $\cH_p$ by
\eqn\label{eq-bbstar-def-1}
	b_{\lambda}(k) \, :=   U e^{ ikx} a_\lambda(k) U^{-1} 
	\; \; \; , \; \; \; \;
	b_{\lambda}^*(k) \, := \,  U e^{- ikx}  a_\lambda^*(k) U^{-1}\,, 
\eeqn 
i.e.,
\eqn
	(Ue^{ ikx}a_{\lambda}(k) \Psi)( p) \, =\, b_{\lambda}(k)(U\Psi)(p )
	\; \; \; , \; \; \; \;
	(Ue^{-ikx}a_{\lambda}^*(k) \Psi)( p) \, =\, b_{\lambda}^*(k)(U\Psi)(p ) \,,
\eeqn 
for $\Psi \in \mathcal{H}$. Obviously, the operator-valued distributions $b_{\lambda}(k)$ and $b_{\lambda}^{*}(k)$ commute with $P_{tot}$. Thus, the operators
$b_\lambda^{(*)}(f):=\int  b_\lambda^{(*)}(k) \, \widehat f(k) dk $ map the fiber spaces $\cH_p$ 
to themselves, for any test function $f$.
The fact that these operators satisfy the usual canonical commutation relations is obvious. The Fock space constructed from the operators $b_\lambda^{(*)}(f)$, $f \in L^2(\R^3 \, \times \, \{+,-\} \, )$, and the vacuum vector $\Omega$ is denoted by $\Fo^{b}$.

From abstract theory, 
the  fiber operators $\h(p)$, $p \in \mathbb{R}^3$, are nonnegative self-adjoint operators acting on $\cH_p\cong\Fo^b$.  
Their explicit form is determined  in the next section.
We define $E(p)=\inf\spec \h(p)$, for all $p \in \mathbb{R}^3$, and
\eqn
	\cS \, :=  \, \big\{ \, p \in \mathbb{R}^3 \, \big| \, |p| \le \frac13 \, \big\}.
\eeqn
Making use of approximate ground states, $\PhiN^\rho(p)$, $\rho > 0$, 
(dressed by a cloud of soft photons with frequencies below $\rho$) of the operators $\h(p)$, which will be defined in \eqref{eq-Phsig-def-0-1}, we introduce
a family of maps  $\JopN^\rho: L^2(\R^3)\mapsto \cH$,
from the space $L^2(\R^3)$ of square-integrable one-particle wave functions, $u$,
to  a subspace of {\em dressed one-electron states},
$\widehat u\, \PhiN^\rho$, as  
\eqn\label{eq-JopN-def-1}
	\JopN^\rho( \, u \, ) (x)\, :=  (U^{-1}\, \id_{\cSd}	\, \widehat u\, \PhiN^\rho)(x)\, =
	(2\pi)^{-3/2}\int_{} dp \, \widehat u(p) \, e^{ix(p-\Pf)} \, \id_{\cSd}(p)
	\, \PhiN^\rho(p)  \,,
\eeqn
where $\id_{\cSd}$ is a smooth approximate characteristic function of the set $S_\mu := (1-\mu)\cS\subset\cS \subset \mathbb{R}^3$, ($0<\mu<1$).

In this paper we study the time evolution of  one-electron states, 
$\JopN^\rho (u_0^\e)$,
where $u_{0}^{\epsilon}$ is a slowly varying one-particle wave function,
dressed by an {\em infrared cloud} of photons with frequencies  below $\rho$.
More precisely, we study solutions of the Schr{\"o}dinger equation
\eqn\label{fullSE}
	i\partial_t \Psi(t) \, = \, \HV \, \Psi(t)
     	\; , \; \; \; \; \text{with} \; \; 	\Psi(0) \, = \, \JopN^\rho (u_0^\e)  \, .
\eeqn
The key idea is to relate the solution
$\Psi(t) \, = \, e^{-it \HV}\JopN^\rho (u_0^\e)$ of this Schr\"odinger equation to the solution of the
Schr\"odinger equation
\eqn\label{effSchrodinger-1-0}
i \partial_t  u_t^\e 
\, = \, 
\Hef \, u_t^\e  \; , \;  \; 
	\text{with} \quad 
	u_{t=0}^\e = u_0^\e \ ,
\eeqn
corresponding to the one-particle Schr\"odinger operator \eqref{Heff-def-1-0},
where we recall that $\Hef = E(-i\nabla_x)  + \Ve(x)$,
with $E(p)$ as defined above. 
We consider the comparison state
\eqn\label{effdyn-def-0}
	\JopN^\rho( u_t^\e )  \; \in \, \cH \,,
\eeqn
where $ u_t^\e 
:=  e^{-it \Hef}  u_0^\e$
is the solution of \eqref{effSchrodinger-1-0}, and
show that $\Psi(t)$ remains close to $\JopN^\rho( u_t^\e )$, for a long time.
The choice of initial data satisfying 
\eqn\label{u0}
	\| u_0^\e \|_{L^2(\R^3)} =1\ \quad \mbox{and}\ \quad\|\nabla u_0^\e \|_{L^2(\R^3)} \leq \e^\kappa
	\;\; , \; \;  0 \leq \kappa< \frac13 \,,
\eeqn guarantees that
$\widehat u_t^\e $ 
remains concentrated in $\cS$ during the time scales relevant for this problem,
provided the support of $\widehat u_0^\e$ is contained in $\cS$.

\begin{theorem}\label{thm-main-1-0}
Let $0< \e < 1/3$, $0 \le \kappa <1/3$ and assume that $u_0^\e \in L^2(\R^3)$ obeys \eqref{u0}. 
Assume, furthermore, that $V \in L^\infty(\R^3; \R)$ is such that $\widehat V \in
L^1(\R^3)$ and that $\widehat V$ is supported in the unit ball,
\eqn\label{eq-Vpot-hyp-0}
	\supp(\widehat V) \, \subset \, \big\{ \, k \in \mathbb{R}^3 \, | \,  |k| \le 1 \, \big\} \, .
\eeqn
 Let $0 < \delta < 2(\frac13 - \kappa)$, and choose $\rho = \rho_\e : = \e^{\frac23-\delta}$.\\ 
 Then there exists 
 %{\color{Red}
 $0<\alpha_\delta\ll1$
 %}
such that, for all $0 \le \alpha \le \alpha_\delta$,
the bound
\eqn\label{eq-main-thm-difference-0-1}
	\| \, e^{-it H^V} \,\JopN^{\rho_\e} (\,  u_0^\e \, )
	\, - \,\JopN^{\rho_\e} (\,  e^{-itH_{{\rm eff}}} \, u_0^\e \, ) \, \|_{\cH}
	\ \leq \  
        C_\delta \, \big( \,         
        \e^{\frac{1}{3} - \frac{\delta}{2} + \kappa} \: t 
            \: + \: \e^{\frac{4}{3} - \frac{\delta}{2}} \: t^2 
	\, \big) \, ,
\eeqn
holds for all times $t \geq 0$.
In particular, for all $0 \leq t \leq \e^{-2/3}$, we have that
\eqn\label{eq-main-thm-difference-0-1,1}
	\| \, e^{-it H^V} \,\JopN^{\rho_\e} (\,  u_0^\e \, )
	\, - \,\JopN^{\rho_\e} (\,  e^{-itH_{{\rm eff}}} \, u_0^\e \, ) \, \|_{\cH}
	\ \leq \  
        C_\delta \; \e^{\frac{1}{3} - \frac{\delta}{2} + \kappa} \; t \, .
\eeqn
\end{theorem}

\begin{remark}
We note that for this result, the regularity properties of the dressed electron states are crucial,
as described in \eqref{Phi}, below.
\end{remark}

\begin{remark}  Theorem \ref{thm-main-1-0} implies that, for all 
$\delta' >0$ such that $\delta' < \frac13 - \frac{\delta}{2} + \kappa$ 
\eqn
	\| \, e^{-it H^V} \,\JopN^{\rho_\e} (\,  u_0^\e \, )
	\, - \,\JopN^{\rho_\e} (\,  e^{-itH_{{\rm eff}}} \, u_0^\e \, ) \, \|_{\cH}
	\ \leq \  C_\delta \, \, \e^{\delta'} 
\eeqn
holds for all times $t$ with 
$0 \leq  t  \leq 
\e^{-(\frac{1}{3} - \frac{\delta}{2} + \kappa) + \delta'}$.
\end{remark}

\begin{remark}\label{mass-rem} 
The initial conditions in Theorem \ref{thm-main-1-0}  are chosen such that the initial momentum is $O(\e^\kappa) $. 
The conditions on the external potential imply that the expected force, and, thus, the acceleration, is of order  $O(\e)$.  
Hence, at time $t$, the momentum is of order $O(\e^\kappa) + O(\e t)$, and therefore the action, $E(p) t - E(0) t \approx \frac{1}{2m_{ren}} p^2 t$, is of order $O(\e^{2\kappa} t) + O(\e^2 t^3)$.  
Hence, if $\frac13-\kappa > \frac{\delta}{2}$ and $t\le \e^{-1+\kappa}$, then this term is much larger than the error term in Eq.  \eqref{eq-main-thm-difference-0-1}. 
  
To make this remark more precise, we
define the operator $\widetilde H_{\rm eff} := E(0) + V(\e x)$, and consider the difference between  
$e^{-i t H_{\rm eff}}$ and $e^{-i t \widetilde H_{\rm eff}}$.  
We write $e^{-i t H_{\rm eff}} - e^{-i t \widetilde H_{\rm eff}}$ as the integral of a derivative,
\eqn
	e^{-i t H_{\rm eff}} - e^{-i t \widetilde H_{\rm eff}}=-i \,  \int_0^t ds 
	\, e^{-i(t-s)H_{\rm eff}} \, (E(p)   -E(0))e^{-is \widetilde H_{\rm eff}} \,,
\eeqn 
and use that $E'(0) =0$  so that  $cp^2<E(p)   -E(0)=\frac{1}{2m_{\rm ren}} p^2 (1+o(1))<Cp^2$
(see Proposition \ref{prp-En-properties}, below).
Then, using 
\eqn 
	e^{i s \widetilde H_{\rm eff}} p^2e^{-i s \widetilde H_{\rm eff}}  
	\, = \, 
        p^2 +2\e \, p \cdot (\n V)(\e x)  s+\e^2 \Delta V (\e x) s^2 \, , 
\eeqn 
we find that 
\eqn \label{eff-prop-exp}
	e^{-i t H_{\rm eff}} - e^{-i t \widetilde H_{\rm eff}} 
	\, = \, A_t + O(\e t^2 p)+ O(\e^2 t^3),  
\eeqn
where $A_t = O(t p^2)$. Adding the second and third term on the r.h.s.\
of \eqref{eff-prop-exp} to the error estimated by \eqref{eq-main-thm-difference-0-1,1}, we observe that 
\eqn \label{eff-prop-exp-1,1}
O(t \e^{\frac{1}{3} - \frac{\delta}{2} + \kappa}) + O(\e t^2 p)+ O(\e^2 t^3)
\ = \ 
O\big[t (\e^{\frac{1}{3} - \frac{\delta}{2} + \kappa} +
\e^{\frac{1}{3} + \kappa} + \e^{\frac{2}{3}} ) \big]
\ = \ 
O(t \e^{\frac{1}{3} - \frac{\delta}{2} + \kappa}) ,
\eeqn
provided that $0 < t \leq \e^{-2/3}$. Assuming that $A_t$ is not
only bounded above by $O( t p^2)$, but is actually of order 
\eqn \label{eff-prop-exp-1,2}
\| A_t \, u_0^\e \| \ \geq \ C \, t \, \e^{2\kappa} \, ,
\eeqn
with $C \equiv C(u_0^\e, m_{\rm ren}) > 0$ depending on
the initial data and on the renormalized mass $m_{\rm ren}$,
we can compare this contribution to \eqref{eff-prop-exp-1,1} and observe
that
\eqn \label{eff-prop-exp-1,3}
%\frac{ \| \, e^{-it H^V} \,\JopN^{\rho_\e} (\,  u_0^\e \, )
%	\, - \,\JopN^{\rho_\e} (\,  A_t \, u_0^\e \, ) \, \|_{\cH} }
%{\| A_t \, u_0^\e \|}
%\ \leq \  
%O(\e^{\frac{1}{3} - \frac{\delta}{2} - \kappa}) \: \| A_t \, u_0^\e \| ,
%\nonumber\\
%{\bf should\;be:}\quad
	\frac{ \| \, e^{-it H^V} \,\JopN^{\rho_\e} (\,  u_0^\e \, ) 
	\, - \,  \JopN^{\rho_\e} (\,  e^{-it\widetilde H_{{\rm eff}}} \, u_0^\e \, )
	\, - \, \JopN^{\rho_\e} (\,  A_t \, u_0^\e \, ) \, \|_{\cH} }
	{\| A_t \, u_0^\e \|}
\ \leq \  
O(\e^{\frac{1}{3} - \frac{\delta}{2} - \kappa}) 
\eeqn
provided $\e^{-2\kappa} \leq t \leq \e^{-2/3}$. Thus our estimate
allows us to separate the main contribution of the dynamics from
the error terms on a suitable time scale.
\end{remark}

 \subsection{Outline of proof strategy}
To prove Theorem \ref{thm-main-1-0}, we introduce an
{\em infrared regularized} version of the
model defined by \eqref{He-def-0}, \eqref{Hn-def-0}, obtained by restricting the
integration domain in the quantized electromagnetic vector potential \eqref{Af-def-0-1}
to the region $\{ \sigma \le |k| \le 1 \}$, for an arbitrary infrared cutoff $\sigma>0$.
Thereby, we obtain infrared regularized Hamiltonians $\He$ and $\Hn$,
as well as an infrared regularized family of maps $\Jop^\rho$
corresponding to $\JopN^\rho$.

We note that,
unlike $H(p)$,  the
infrared cut-off
fiber Hamiltonian $\Hn(p)$ has a ground state
$\Psig(p)\in\cH_p\cong\Fo$, for every $p\in\cS$ and for $\sigma>0$,
but $\Psig(p)$ does not possess a limit in $\cH_p\cong\Fo$,  as $\sigma \searrow 0$,
%{\color{Red}
when $p\neq0$.
%}
In particular, we expect that the number of photons in the state
$\Psig(p)$ diverges, as $\sigma\searrow0$, (thus the lack of convergence of $\Psig(p)$ in $\Fo$).
This is a well-known aspect of the {\em infrared problem} in QED, \cite{chfr,cfp1,cfp2,fr1,pi}.
It is remedied by
applying a dressing transformation, $W_{\nablE(p)}^{\s,\rho}$,
defined in
\eqref{eq-Phsig-def-0-1}, below,  to $\Psig(p)$,
where $\Eg(p)=\inf{\rm spec}\Hn(p)$.
The resulting vector, 
%{\color{Red}
$\Phsig^\rho(p)  := W_{\nablE(p)}^{\s,\rho}  \Psig(p)$,
%}
describes an {\em infraparticle} (or {\em dressed electron) state} containing
infrared photons
with frequencies in $[\sigma,\rho]$.
As $\sigma\searrow0$, the limit 
\eqn\label{Phi}
	\Phi^\rho(p)=\lim_{\sigma\rightarrow0}\Phsig^\rho(p)
\eeqn
exists in $\Fo$, for all $p\in\cS$; see Proposition \ref{prp-Phsig-Holder-1}.
This allows us to construct the map  $\JopN^\rho$ as
the limit of the maps $\Jop^\rho$, as $\sigma\searrow0$.
Note that, while $\Psi_\sigma(p)$ does not converge in $\Fo$ 
%{\color{Red}
as $\sigma\searrow0$ when $p\neq0$,
%} 
we have that $\lim_{\sigma \searrow 0} E_\sigma(p) = E(p)$.

We note that $\Phsig^\rho(p)$ is the ground state eigenvector of the
fiber Hamiltonian
\eqn\label{eq:Kn2}
	\Kn^\rho(p) \, := \, W_{\nablE(p)}^{\s,\rho} \, \Hn(p) \, (W_{\nablE(p)}^{\s,\rho})^*
\eeqn
which is obtained by applying to $\Hn(p)$  the {\em Bogoliubov transformation}
corresponding to the dressing transformation $W_{\nablE(p)}^{\s,\rho}$.

In Theorem \ref{thm-main-1-1}, below, we prove that an estimate similar to
\eqref{eq-main-thm-difference-0-1}
is satisfied for the infrared regularized model; namely,
\eqn\label{eq-main-thm-difference-0-3}
	\| \, e^{-it H^V_\sigma} \,\Jop^\rho (\,  u_0^\e \, )
	\, - \,\Jop^\rho (\,  e^{-itH_{{\rm eff},\sigma}} \, u_0^\e \, ) \, \|_{\cH}
	\, \leq \,  C_\delta \, (1+\ln(\rho^{-1})) \,  \e^{\frac23-\delta} \, t
	+ C \, \alpha^{\frac12} \, \rho^{\frac12} \,   t \, (\e^\kappa +\e t )  \,  ,
\eeqn
holds uniformly in the infrared cutoff $\sigma$  and the cut-off  $\rho >\s$.
This result crucially uses the regularity properties of the dressed electron states
$\Phsig^\rho(p)$,
which allow us to take advantage of the fact that $V_\e$ is slowly varying.
An additional key ingredient
is the bound $\|(\Hn(p)-\Kn^\rho(p))\Phsig^\rho(p)\|_{\Fo}\,
\leq \, C  \alpha^{\frac12} \, \rho^{\frac12}|p|$,
for $p\in\cS$, proven in
Appendix \ref{ssec-Dop2-2}.
In \eqref{eq-main-thm-difference-0-3} we take $\rho=\rho_\e := \e^{\frac23-\delta}$ and absorb $\ln(\rho^{-1}) $ into $\e^{\frac23-\delta}$.

In Section \ref{sec-IRlim-1}, we control the limit $\sigma\searrow0$, thus concluding the
proof of Theorem \ref{thm-main-1-0}.
This requires
control of the radiation emitted by the electron
due to its acceleration in the external potential $\Ve$, in the limit
$\sigma\searrow0$.

\subsection*{Acknowledgements}
The authors are very grateful to Herbert Spohn and Stefan Teufel for
pointing out a somewhat serious problem in an earlier version of this paper and for  
suggesting to us a solution. J.Fa., I.M.S., and T.C. are  grateful to J.Fr. for hospitality at ETH Z{\"u}rich. T.C. thanks I.M.S. for hospitality at the University of Toronto.  The research of I.M.S. has been supported by NSERC under Grant NA 7901. T.C. has been supported by the NSF under grants
DMS-0940145, DMS-1009448, and DMS-1151414 (CAREER).

$\;$ \\

\section{Infrared cut-off and construction of $\PhiN^\rho(p)$}

As noted in the introduction, we analyze the original dynamics by first imposing
an infrared (IR) cut-off, and controlling the dynamics generated by the resulting Hamiltonian.
Thus, we  define the IR regularized Hamiltonian
\eqn\label{eq-He-def-1}
	\He \, = \, \Hn \, + \, \Ve(x)\otimes\1_f \, ,
\eeqn
where
\eqn\label{eq-Hn-def-1}
	\Hn \, := \, \frac12( \,- i\nabla_x \otimes \1_f \,  + \, \g \Af(x) \, )^2
	\, + \, \1_{el} \otimes H_f \,
\eeqn
is the generator of the dynamics of a single, freely moving
non-relativistic electron minimally coupled to
the electromagnetic radiation field.
In \eqref{eq-Hn-def-1},
\eqn
	\Af(x) \, = \,
	\sum_{\lambda}\int_{\sigma \le |k| \le 1} \, \frac{ dk }{|k|^{1/2}} \,
	\{ \,  \e_\lambda(k) \,  e^{ikx} \otimes a_\lambda(k) \, + \, h.c. \, \} \label{eq:defAf}
\eeqn
denotes the quantized electromagnetic vector potential
with an infrared and ultraviolet cutoff corresponding to
$\sigma \le |k| \le 1$.
Since $V \in L^\infty( \mathbb{R}^3 )$ is a bounded operator,
$D(H_\sigma^V) = D( H_\sigma ) = D( -\Delta_x \otimes \1_f + \1_{el} \otimes H_f )$.
%{\color{Red}
The results in this paper are proven for sufficiently small values of the finestructure constant,
$0<\alpha\ll1$.
%}

The Hamiltonian $\Hn$ is also translation invariant and, similarly to  $\h$, can be represented 
as the fiber integral
\begin{equation}\label{H-fib-deco-s}
U\Hn U^{-1} = \int_{\mathbb{R}^3}^{\oplus} \Hn(p) dp,
\end{equation}
over the spectrum of $\Ptot$,  defined on the fiber integral $\int^{\oplus}dp \, \cH_p$, with   fibers $\cH_p \cong\Fo^b$. 
 The decomposition \eqref{H-fib-deco-s}  is equivalent to
 \eqn \label{H-fib-deco-s'}	(U\Hn \Psi)(p) \, = \, \Hn(p)  (U\Psi)(p).
\eeqn
Again, by abstract theory, the  fiber Hamiltonians $\Hn(p)$, $p \in \mathbb{R}^3$, are self-adjoint operators on 
$\cH_p\cong\Fo^b$. 
Written in terms of the creation- and annihilation operators on the fiber space, they are given by
\begin{equation}\label{Hp}
	\Hn(p) =   \frac{1}{2}  \big (p - P_f^b  - \, \g \Af^b)^2 + H_f^b
\end{equation}
where 
\eqn
	H_f^b \, := \, \sum_{\lambda}\int dk \, |k| \, b_\lambda^*(k) \, b_\lambda(k)
	\; \; \; , \; \; \; \;
	P_f^b \, := \, \sum_{\lambda}\int dk \, k \, b_\lambda^*(k) \, b_\lambda(k)
\eeqn
and
\begin{equation}\label{Achi}
	\Af^b \, := \,
	\sum_{\lambda}\int_{\sigma \le |k| \le 1}    \, \frac{dk}{|k|^{1/2}} \,
	  \,\{ \,   \e_\lambda(k) \, b_\lambda(k) \, + \, h.c.\, \}. 
\end{equation}
Henceforth, we will drop the superscripts $"b"$ from the notation.

While $H(p)$ has a ground state only for $p=0$,
it is proven in \cite{bcfs2,ch} that, for
$p\in\cS:= \{ p \in \mathbb{R}^3 | |p| \le 1/3\}$
and $\sigma>0$,  $\Hn(p)$ has  a non-degenerate 
(fiber) ground state.  This motivates the introduction of the cut-off.
Properties of the fiber ground state energy, $\Eg(p)\, = \, \inf\spec \Hn(p)$, 
are given in the following proposition proven in \cite{bcfs2,ch,cfp1,cfp2}:

\begin{proposition}\label{prp-En-properties}
%{\color{Red}
There exists a  constant $0<\alpha_0 \ll1$ such that for all $0 < \alpha \leq \alpha_0$, 
the
%}
infimum of the spectrum of the fiber Hamiltonian,
\eqn
	\Eg(p) \, = \, \inf\spec \Hn(p),
\eeqn
satisfies:
\begin{enumerate}
\item 
For any $\sigma>0$, $\Eg\in C^2(\cS)$, and for all
$p \, \in \, \cS \, = \, \big\{ \, p \in \mathbb{R}^3 \, | \,  |p| \le \frac13 \, \big \} \,,$
 $\Eg(p)$ is a simple eigenvalue.
\\

\item
There exists a constant $c < \infty$
such that, for any $p\in\cS$ and $\sigma \geq 0$, we have that
\eqn
	  |\nabla_{p} \Eg(p) - p | \, \le \,  c \, \alpha \, |p| \, ,
 \quad \mbox{and} \quad 	1 - c \, \alpha \, \le \, \partial_{|p|}^2 \Eg(p) \, \le 1\, . \label{eq:boundd2E}
\eeqn
\\

\item
The following limit exists in $C^2(\cS)$
\eqn
	\lim_{\sigma\searrow0}\Eg( \, \cdot \,) \, = \, E( \, \cdot \,).
\eeqn

\end{enumerate}
\end{proposition}

%*******************

We let $\Psig(p)\in\Fo$, with $\|\Psig(p)\|_\Fo=1$,
denote the normalized fiber ground state corresponding to $\Eg(p)$,
\eqn
	\Hn(p)\Psig(p) \, = \, \Eg(p) \, \Psig(p) \,,
\eeqn
for $p\in\cS$.
For $0<\sigma<\rho \le 1$ and $p\in\cS$,  we introduce the Weyl operators
\eqn\label{eq-Weylop-def-1}
	W_{\nablE(p)}^{\s, \rho}
	\, := \,   \exp\Big[\alpha^{\frac12}\sum_\lambda\int_{\sigma\le|k|\le\rho} dk \,
	\frac{\nablE(p)\cdot\e_\lambda(k)b_\lambda(k) -h.c.}{|k|^{1/2}(|k|-\nablE(p)\cdot k)}\Big] \,,
\eeqn
with $\nablE(p)  \equiv  \nabla_p \Eg(p)$, which are unitary on $\Fo$, for $\sigma>0$.
Moreover, we define dressed electron states
\eqn\label{eq-Phsig-def-0-1}
	\Phsig^\rho(p) \, := \, W_{\nablE(p)}^{\s, \rho} \, \Psig(p) \,.
\eeqn
For $p\in\cS$, we define the Bogoliubov-transformed fiber Hamiltonians
\eqn \label{eq:Kn}
	\Kn^\rho(p) \, := \, W_{\nablE(p)}^{\s, \rho} \, \Hn(p) \, (W_{\nablE(p)}^{\s, \rho})^* \, .
\eeqn 
It is convenient to define $\Kn^\rho(p) := \Hn(p)$, 
for $p \in \mathbb{R}^3 \setminus \cS$.

The  dressed electron states $\Phsig^\rho(p)$, for $p\in\cS$,
are the ground states of  the Bogoliubov-transformed fiber Hamiltonians $\Kn^\rho(p)$, defined in \eqref{eq:Kn}, i.e.,
\eqn
	\Kn^\rho(p) \, \Phsig^\rho(p) \, = \, \Eg(p) \, \Phsig^\rho(p) \,.
\eeqn
The properties of these states
are described in the following proposition
\begin{proposition}\label{prp-Phsig-Holder-1}
For any $p\in \cS$, $0<\rho\le1$, and for sufficiently small values of the finestructure constant $0<\alpha\ll1$,
the ground state eigenvector $\Phsig^\rho(p)$ 
satisfies:
\begin{enumerate}
\item
The strong limit
\eqn
	\Phi^\rho(p) \, := \, \lim_{\sigma\rightarrow0}\Phsig^\rho(p)
\eeqn
exists in $\Fo$.
\\
\item
For $\holder <\frac23$, the vectors $\Phsig^\rho(p)$ are $\holder $-H\"older continuous in $p$,
\eqn
	\sup_{p,q\in\cS}\frac{\| \, \Phsig^\rho(p)-\Phsig^\rho(q) \, \|}{|p-q|^\holder } 
	\, \le \, C(\holder ) \;
	%{\color{Red}
	\ln\frac1\rho
	%} 
	\, < \, \infty  \, ,
\eeqn
uniformly in $\sigma$, with $0 \le \sigma < \rho \le 1$.
\end{enumerate}
\end{proposition}
The proof of $\holder$-H\"older continuity for $\holder<\frac23$ is given in
Section \ref{sec-HolderContin-1}; (see also \cite{cfp1,cfp2,pi} for earlier results
covering the range $\holder <\frac14$, in the case where $\rho=1$).

For arbitrary $u\in L^2(\R^3)$ (with Fourier transform denoted by $\widehat u$),
we define the linear map
\eqn \label{eq:defJrhosigma}
	\Jop^\rho \, : \,
	u & \mapsto & (2\pi)^{-3/2}  \int_{\cS} dp \, \widehat u(p) \,  \, e^{ix(p-\Pf)} \, \id_{\cSd}(p) \, \Phsig^\rho(p)  \,,
\eeqn
where $x$ is the electron position, $\id_{\cSd}$ is a smooth approximate characteristic function
of the set
\eqn
	\cSd \, := \, (1-\mu) \, \cS\subset\cS \subset \mathbb{R}^3, \label{eq:defSmu}
\eeqn
and $0 < \mu < 1$.
 Note that
$\Jop^\rho \, : \, L^2(\R^3)  \rightarrow  \cM\, \subset \, \cH$, where
\eqn\label{eq-Mspace-def-1}
	\cM \, := \, \Big\{  \, (2\pi)^{-3/2} \int_{\R^3} dp \, \widehat u(p) \, e^{ix(p-\Pf)} \, \id_{\cSd}(p) \, \Phsig^\rho(p)
	 \,
	\Big| \, u \in L^2(\R^3) \, \Big\},
\eeqn
 the subspace of vectors in $\cH$ supported on the one-particle shell of
the operator $\int_{\cS}^\oplus dp \, \Kn^\rho(p)$. We also note that in \eqref{eq-Mspace-def-1}
we do not require that ${\rm supp}(\widehat u)\subset\cSd$; instead,
we cutoff $\widehat u$ outside the region $\cSd$ by multiplying it by $\id_{\cSd}$.

Furthermore, we introduce the one-particle Schr\"odinger operator
\eqn\label{eq-Heff-def-1-0}
	\Heff \, := \, \Eeff(-i\nabla_x)  \, + \, \Ve(x) \,.
\eeqn
Here, the kinetic energy operator is defined by
\eqn
	\Eeff(p) \, := \, \Eg(p) \; , \;\; p\in\cS \,,
\eeqn
and suitably extended to $p \in \mathbb{R}^3\setminus \cS$. Note that the restriction of $\Eeff$ to
$\cS$ is twice continuously differentiable, $\Eeff|_\cS \in C^2(\cS)$;
%{\color{Red}
see Proposition \ref{prp-En-properties}.
%}

As a first step  towards 
proving Theorem \ref{thm-main-1-0}, we prove the following result.

\begin{theorem}\label{thm-main-1-1}
Under the conditions of Theorem \ref{thm-main-1-0},   there
exists
$ \alpha_\delta >0$ such that, for all $0 \leq \alpha \leq\alpha_\delta$,  the bound \eqref{eq-main-thm-difference-0-3}
holds uniformly in the infrared cutoff $\sigma  >0 $ 
and the cutoff $\rho >\sigma$.
\end{theorem}
\begin{proof}
Our proof  makes crucial use of the
properties of the fiber ground state energy $\Eg(p)$ and of the
corresponding dressed electron states $\Phsig^\rho(p)$, for $p\in\cS$,
 given in Propositions \ref{prp-En-properties} and \ref{prp-Phsig-Holder-1} above.
 We define  the operator $\Kn^\rho$ acting on $\cH$,
\eqn\label{eq-Knrho-def-1}
	\Kn^\rho \, := \, \int^\oplus \Kn^\rho(p) \, dp,
\eeqn
and the perturbed operator $\Kn^V \, := \Kn^\rho+\Ve $.
Note that  the operator $\Kn^\rho$ has the  property that
\eqn\label{KHeff}
	\Kn^\rho\Jop^\rho =\Jop^\rho \Eeff (-i\n).
\eeqn

We  write the difference on the LHS of \eqref{eq-main-thm-difference-0-3}
as the integral of a derivative,  substitute $\Hn^V\rightarrow\Hn^V-\Kn^V+\Kn^V$ inside the integral and
group terms suitably to obtain 
\eqn\label{eq-maindiffterm-1}
	e^{-it H^V_\sigma} \,\Jop^\rho (\,  u_0^\e \, )
	\, - \,\Jop^\rho (\,  e^{-itH_{{\rm eff},\sigma}} \, u_0^\e)
	&=&
	-i \, e^{-it\He} \int_0^t ds \, e^{is\He} \, ( \, \Hn^V \Jop^\rho( u_s^\e )   
	-\Jop^\rho (H_{{\rm eff}}  u_s^\e ) \, )
	\nonumber\\
	&=:&\Dop^1(t) \, + \, \Dop^2(t) \, ,
\eeqn
where  $u_s^\e := e^{-isH_{{\rm eff},\sigma}} \, u_0^\e$ and 
\eqn
	 \Dop^1(t) & := &
	-i \, e^{-it\He} \int_0^t ds \, e^{is\He} \, \big( \, \Hn-\Kn^\rho \, \big) \,  \Jop^\rho (\,  u_0^\e \, )
	 \,,
\eeqn
where we have used the cancellation of $V$ in $\Hn^V-\Kn^V=\Hn-\Kn^\rho$, and 
\eqn
	\Dop^2(t) & := &
	-i \, e^{-it\He} \int_0^t ds \, e^{is\He} \, \big( \,
	\Kn^V \Jop^\rho( u_s^\e )   -\Jop^\rho (H_{{\rm eff}}  u_s^\e )
	\, \big).
	\nonumber
\eeqn
The first term on the r.h.s. of \eqref{eq-maindiffterm-1} accounts for the radiation of infrared photons,
%due to the motion of the dressed electron in the external potential, 
while the second term accounts for the influence of the external potential $\Ve$
on the  full QED dynamics  
$ \Psi(t) \, = \, e^{-it \He}\Jop^\rho (u_0^\e)  $, as compared to  the
effective Schr\"odinger evolution $e^{-itH_{{\rm eff},\sigma}}u_0^\e$.

Using the direct integral decomposition, we obtain
\eqn\label{eq-phi1-aux-0-1}
	%{\color{Red}
	\Dop^1(t)  \, = \, 
	-i \, (2\pi)^{-\frac32} \, e^{-it\He} \int_0^t ds \, e^{is\He} \,
	\int_{\cS} dp \, \widehat u_s^\e(p) \, e^{i(p-\Pf)x} (\Hn-\Kn^\rho)(p) 
	\, \id_{\cSd}(p) \, \Phsig^\rho(p)
	\,,
	%} 
\eeqn
so that
\eqn\label{eq-phi1-aux-0-2}
	\|\Dop^1(t)\|_{\cH}&\leq&
	\sup_{p\in\cS} 	\Big\{	\frac{1}{|p| } \, \| (\Hn-\Kn^\rho)(p) \, \Phsig^\rho(p)\|_{\Fo} \, \Big\}
	 \, \int_0^t \| \,\n  u_s^\e  \, \|_{L^2(\R^3)} \, ds \,.
\eeqn
%{\color{Red}
We note that thanks to  $\id_{\cSd}$ in \eqref{eq-phi1-aux-0-1},
which cuts off  the tail of $u_s^\e$ outside of $\cS_\mu$,  
the supremum in \eqref{eq-phi1-aux-0-2} can be taken only for $p\in\cS_\mu$, respectively, $\cS$.
%}

In Appendix \ref{ssec-Dop2-2}
we prove the
following
key result:
\eqn\label{HK-est}
	\sup_{p\in\cS}
	\Big\{\frac{1}{|p| } \, \| (\Hn-\Kn^\rho)(p) \, \Phsig^\rho(p)\|_{\Fo} \, \Big\}
	\, \leq \, C  \alpha^{\frac12} \, \rho^{\frac12} \,,
\eeqn
{\em uniformly in} $\sigma\geq0$.
Furthermore,
we have the estimate
\eqn\label{eq-Schrod-mainbd-1}
	\int_0^t   \,
	\| \,   \nabla  u_s^\e  \, \|_{L^2(\R^3)} \, ds
	\, \leq \, C \,  \, t \, (\e^\kappa +\e t ) \, ,
\eeqn
as shown below in \eqref{eq-nabluL2-est-0,9}--\eqref{eq-nabluL2-est-1}, by using the condition $\|\nabla u_0^\e \|_{L^2(\R^3)} \leq \e^\kappa$ on  $ u_0^\e$, and
the fact that the potential $V$ satisfies \eqref{eq-Vpot-hyp-0}. We obtain
\eqn\label{phi1-est}
	\|\Dop^1(t)\|_{\cH} \, \le \, C \,  t \, (\e^\kappa +\e t )  \, \alpha^{\frac12} \, \rho^{\frac12}  ,
\eeqn
which yields the second contribution to the r.h.s.\ of \eqref{eq-main-thm-difference-0-3}.

For the second term on the r.h.s. of \eqref{eq-maindiffterm-1}, using the fiber decomposition and the equation
$\Kn^\rho(p) \, \Phsig^\rho(p) \, = \, \Eg(p) \, \Phsig^\rho(p) $, we have that
\eqn\label{phi2}
	\Dop^2(t) 
	&=&
	-i \, e^{-it\He} \int_0^t ds \, e^{is\He} \, \big( \, \Ve \Jop^\rho( u_s^\e )   -\Jop^\rho (\Ve  u_s^\e ) \, \big) \,.
\eeqn
Let $\| \Phi\|_{C^\theta(\cS)}:=\sup_{p,q\in\cS}\frac{\| \, \Phi(p)-\Phi(q) \, \|}{|p-q|^\holder }$.

In \eqref{psis}--\eqref{eq-Duppbd-2,1} below, we prove  an estimate of the form
\eqn\label{phi2-est}
	\|\Dop^2(t)\|_{\cH}
	&\leq&t \, C  \, \|\widehat{|\nabla|^\Holdexp\Ve}\|_{L^1(\R^3)} \,
	(1 \, + \, \| \Phsig^\rho\|_{C^\theta(\cS)} ) \,,
\eeqn
for $\theta<\frac23$. The key point here is
that the $\theta$-H\"older continuity of the fiber ground state $\Phsig^\rho(p)$
enables us to gain a $\theta$ derivative of the potential, yielding
$\|\widehat{|\nabla|^\Holdexp\Ve}\|_{L^1(\R^3)} \leq C \e^\theta$.
Using  the $\holder $-H\"older continuity of $\Phsig^\rho(\,\cdot\,)$,
which holds uniformly in $\sigma$, with $0 < \sigma < \rho$,
and the fact that
\eqn\label{eq-Vpot-hyp-1}
	\|\widehat{|\nabla|^\Holdexp V}\|_{L^1(\R^3)}
	\, \le \, \gamma \, ,\quad
	\textrm{where}
	\quad	\gamma \,:= \, \|\widehat V(k)\|_{L^1} \, < \, \infty \,,
\eeqn  (see \eqref{eq-Vpot-hyp-0}) and using $\| \Phsig^\rho\|_{C^\theta(\cS)}\leq C_\delta \, (1+ \ln(\rho^{-1}))$,
which we prove in Proposition \ref{prop:Holder},  we arrive at
\eqn\label{phi2-est-fin}
	\|\Dop^2(t)\|_{\cH} &\leq& C_\delta  \,t \, \e^\theta \, (1+ \ln(\rho^{-1})),
\eeqn
which yields the first term on the RHS of \eqref{eq-main-thm-difference-0-3}.
\end{proof}

%*******************

\prf[Proof of   \eqref{eq-Schrod-mainbd-1}]
To verify \eqref{eq-Schrod-mainbd-1},
a simple calculation shows that
\eqn \label{eq-nabluL2-est-0,9}
	\nabla  u_s^\e 
	&=& e^{- i s \Heff } \, \nabla u_0^\e \, - i  \, \int_0^s \, d v \, e^{ - i v \Heff } 
	\, \nabla V_\e(x)  \, e^{ - i ( s - v ) \Heff } \, u_s^\e.
\eeqn
Using that $\| \nabla u_0^\e \|_{L^2} \le \e^\kappa$, and that
\eqn
\| \nabla V_\e \|_{L^\infty}
\, = \, \|\widehat{\nabla \Ve}\|_{L^1}
	\, \leq \, \gamma \, \e \,,
\eeqn
we conclude that
\eqn\label{eq-nabluL2-est-1}
	\| \, \nabla   u_s^\e  \,\|_{L^2}
	&\leq&C \,   (\e^\kappa +\e s )  \,,
\eeqn
and thus,
  \eqref{eq-Schrod-mainbd-1}.
\endprf

%*******************

\begin{proof}[Proof of \eqref{phi2-est}]
 In what follows we use the notation 
 $$(U \Psi)(p) = \widehat\Psi(p)\  \quad \mbox{and}\  \quad (U^{-1}\Phi)(x) = \Phi^\vee(x) .$$  
 We define
\eqn\label{psis}
	\Aop_s & := &
	\Ve \Jop^\rho(  u_s^\e  )  - \Jop^\rho  ( \Ve  u_s^\e ). 	
\eeqn	
Using the definition of $\Jop^\rho$ and computing the Fourier transform, we find that
\eqn	\widehat\Aop_s (p)&=& (2\pi)^{-3/2} 
	\int_{\R^3} dq \,  \widehat\Ve(p-q) \, \widehat u(s,q) \,
	\big( \, \id_{\cSd}(q) \Phsig^\rho(q) \, - \, \id_{\cSd}(p)  \Phsig^\rho(p) \, \big).
\eeqn
By  relations \eqref{phi2} and \eqref{psis} and the unitarity of the generalized Fourier transform 
we have that
\eqn\label{eq-Duppbd-2}
	\|\Dop^2(t)\|_{\cH}  \leq \int_0^t ds \, \|  \Aop_s\|_{L^2_x\otimes\Fo}
	 = \int_0^t ds \, \|  \widehat{\Aop_s}\|_{L^2_p\otimes\Fo}.
\eeqn
It is important to note that, for any function $f\in L^2(\R^3)$ with ${\rm supp}(f)\subset \cSd$,
\eqn\label{eq-hatVe-supp-1}
	{\rm supp}(\widehat\Ve*f) \, \subset \, \cS ,
\eeqn	
for $\e \le \mu/3$, since we are assuming $\supp(\widehat V)\subset \{ k | |k|Ê\le 1 \}$,
so that $\supp(\widehat \Ve)\subset \{ k | | k | \le \e \}$.
 Since the term in the integrand given by $( \widehat u_s^\e   \id_{\cSd}  \Phsig^\rho)(q)$
is supported in $q\in\cSd$,
so that, by \eqref{eq-hatVe-supp-1}, its convolution with $\widehat \Ve$ has support in $\cS$,
we find
\eqn\label{psisFT-bnd}
	\widehat\Aop_s (p)&=& (2\pi)^{-3/2} 	\1_{\cS}(p)\int_{\R^3} dq \,  \widehat\Ve(p-q) \, \widehat u(s,q)
	\,  (\id_{\cSd}(q)  \,   \Phsig^\rho(q)- \id_{\cSd}(p)  \Phsig^\rho(p))  \, ,
\eeqn
for  $\e \le \mu/3$, where $\1_{\cS}$ is the characteristic function of the set $\cS$.
Inserting $|p-q|^\holder |p-q|^{-\holder }=\1$ into \eqref{psisFT-bnd},
using the definition of $|\nabla|^\Holdexp$ by its Fourier transform
and using that, since $\id_{\cSd}$ is a smooth function,
\eqn
	\sup_{p,q\in\cS}|p-q|^{-\Holdexp}
	\| \big( \,
	\id_{\cSd}(q) \Phsig^\rho(q) \, - \,
	\id_{\cSd}(p)  \Phsig^\rho(p) \, \big)\|_{ \Fo}\le C(1 \, + \, \| \Phsig^\rho\|_{C^\theta(\cS)} ) \,,
\eeqn
we obtain the bound
$\|\hat\Aop_s \|_{L^2_x\otimes\Fo} \le \, C(1 \, + \, \| \Phsig^\rho\|_{C^\theta(\cS)} )\| \,
|\1_{\cS} \widehat u_s^\e |* |\widehat{|\nabla|^\Holdexp\Ve}| \, \|_{L^2(\cS)}$. Next, using  Young's
inequality, $\|f*g\|_{L^r}\leq\|f\|_{L^1}\|g\|_{L^r}$, we find that
\eqn\label{eq-Duppbd-2'}	
	\|\widehat\Aop_s \|_{L^2_x\otimes\Fo}\leq \, C(1 \, + \, \| \Phsig^\rho\|_{C^\theta(\cS)} )
	\, \|\widehat{|\nabla|^\Holdexp\Ve}\|_{L^1(\R^3)} \,
	\sup_{s\in[0,t]}\|\1_{\cS} \widehat u_s^\e \|_{L^2(\R^3)}.
\eeqn
Finally, observing that
\eqn\label{eq-Duppbd-2,1}
	\|\1_{\cS} \widehat u_s^\e \|_{L^2(\R^3)}
\, \leq \,	\| \widehat u_s^\e \|_{L^2(\R^3)}
=	\|   u_s^\e \|_{L^2(\R^3)}
	\, = \,  \|  u_0^\e \|_{L^2(\R^3)} \, = \, 1 \,,
\eeqn
by unitarity of $e^{-it\Heff}$, and using \eqref{eq-Duppbd-2}, we arrive at \eqref{phi2-est}.
\end{proof}

%****************

\section{The limit $\sigma\searrow0$}
\label{sec-IRlim-1}

In this section we remove the infrared cut-off from the evolution.
\begin{proposition}\label{prop-s-lim}
Under the conditions of Theorem \ref{thm-main-1-1},
the strong limits
\eqn
	\label{eq-main-thm-difference-1-1}
	s-\lim_{\sigma\searrow0} e^{-it\He} \,\Jop^\rho (\,  u_0^\e \, ) \, = \,
	 e^{-itH^V} \,\JopN^\rho (\,  u_0^\e \, )
\eeqn
and
\eqn
	\label{eq-main-thm-difference-1-2}
	s-\lim_{\sigma\searrow0}  \Jop^\rho (\,   e^{-it\Heff} \, u_0^\e \, )  \, = \,
	 \JopN^\rho (\,   e^{-it H_{\rm eff}} \, u_0^\e \, )
\eeqn
exist, for arbitrary $|t|<\infty$.
\end{proposition}
\begin{proof} We write
\begin{align}
e^{-itH_\sigma^V} \mathcal{J}_\sigma^\rho ( u_0^\epsilon ) - e^{-itH^V} \mathcal{J}_0^\rho ( u_0^\epsilon ) \, = \, ( e^{-itH_\sigma^V} - e^{-it H^V} ) \mathcal{J}_0^\rho ( u_0^\epsilon ) + e^{-itH_\sigma^V} ( \mathcal{J}_\sigma^\rho - \mathcal{J}_0^\rho ) ( u_0^\epsilon ). \label{eq:a1}
\end{align}
Clearly,
\begin{align*}
\left \| e^{-itH_\sigma^V} ( \mathcal{J}_\sigma^\rho - \mathcal{J}_0^\rho ) ( u_0^\epsilon ) \right \| \, = \, \left \| ( \mathcal{J}_\sigma^\rho - \mathcal{J}_0^\rho ) ( u_0^\epsilon ) \right \| \, \le \, \| u_0^\epsilon \|_{L^2} \sup_{p \in \mathcal{S}_\mu} \left \| \Phi_\sigma^\rho(p) - \Phi^\rho(p) \right \|_{ \mathcal{F} }.
\end{align*}
Thus,
\begin{align*}
\lim_{\sigma \searrow 0} \left \| e^{-itH_\sigma^V} ( \mathcal{J}_\sigma^\rho - \mathcal{J}_0^\rho ) ( u_0^\epsilon ) \right \| \, = \, 0 \, ,
\end{align*}
follows from Proposition \ref{prop:convergence_GS}.

Next, we discuss the first term on the right side of \eqref{eq:a1}. In order to prove that it converges to $0$, as $\sigma \searrow 0$, it suffices to show that $H_\sigma^V$ converges to $H^V$ in the norm resolvent sense; (see \cite[Theorem VIII.21]{rs1}), i.e.,
\begin{align*}
\lim_{ \sigma \searrow 0} \left \| ( H^V_\sigma + i )^{-1} - (H^V + i)^{-1} \right \| \, = \, 0.
\end{align*}
From the second resolvent equation and the fact that $\| (H_\sigma^V + i)^{-1} \| \, \le \, 1$,
it follows that
\begin{align}
	\left \| ( H^V_\sigma + i )^{-1} - (H^V + i)^{-1} \right \| \,
	= \, \left \| ( H^V + i )^{-1} \, Q_\sigma \, (H^V_\sigma + i)^{-1} \right \| \,,
	\label{eq:a2}
\end{align}
where
\begin{align*}
	Q_\sigma \, :=& \, H^V - H^V_\sigma \, =\, \alpha^{\frac 12}
	A_{<\sigma}(x) \cdot v_\sigma + \frac \alpha2 ( A_{<\sigma}(x) )^2 \,,
\end{align*}
and
\begin{align*}
 	v_\sigma \, := \,
	- i \nabla_x + \alpha^{\frac 12} A_{\sigma}(x)
\end{align*}
is the velocity operator.
Here $A_\sigma(x)$ is defined in \eqref{eq:defAf}, and
\eqn\label{eq-Af-def-0-1}
	A_{<\sigma}(x) \, := \,
	\sum_{\lambda}\int_{|k| \le \sigma}
	\, \frac{dk}{|k|^{1/2}} \,
	\{ \,  \e_\lambda(k) \,  e^{-ikx} \otimes a_\lambda(k) \, + \, h.c. \, \} \, .
\eeqn
In order to estimate the norm of $Q_\sigma (H^V+i)^{-1}$, we use the following well-known lemma.
\begin{lemma}\label{lm:b1}
Let $f,g \in \mathrm{L}^2 ( \mathbb{R}^3 \times \{ + , - \} ; \mathcal{B} ( \mathcal{H}_{ \mathrm{el} } ) )$ be operator-valued functions such that $\| ( 1 + |k|^{-1} )^{1/2} f \|, \| ( 1 + |k|^{-1} )^{1/2} g \| < \infty$. Then
\begin{align}
& \| a^\#(f) ( H_f + 1 )^{-\frac 12} \| 
\, \le \, 
\| ( 1 + |k|^{-1} )^{\frac 12} f \|_{L^2}, \label{lm:b1_1} 
\\
& \| a^\#(f) a^\#(g) ( H_f + 1 )^{-1} \| 
\, \le \, 
\| (1 + |k|^{-1})^{\frac 12} f \|_{L^2} \: 
\| (1 + |k|^{-1})^{\frac 12} g \|_{L^2}, \label{lm:b1_2}
\end{align}
where $a^\#$ stands for $a$ or $a^*$.
\end{lemma}
In particular, using the Kato-Rellich theorem, one easily shows that, for $\alpha$ small enough, $D( H^V ) = D( -\Delta_x \otimes I  + I \otimes H_f ) \subset D( H_f )$. Thus, we have that
\begin{align*}
\left \|(H_f + 1) (H^V+i)^{-1} \right \| \, \le \, C,
\end{align*}
which when combined with Lemma \ref{lm:b1} yields
\begin{align}\label{eq:a3}
\left \|\frac \alpha2 ( A_{<\sigma}(x) )^2 (H^V+i)^{-1}  \right \| \, \le \, C \, \alpha \, \sigma.
\end{align}
Likewise one verifies that
\begin{align}
	& \left \| \alpha^{\frac 12} A_{<\sigma}(x) \cdot
	v_\sigma
	(H^V+i)^{-1}  \right \| \, \le \, C \, \alpha^{\frac12} \, \sigma^{\frac12}, \label{eq:a4}
\end{align}
since $0\leq v_\sigma^2\leq H^V+\|V\|_{L^\infty}$ is bounded relative to $H^V$.
Estimates \eqref{eq:a3} and \eqref{eq:a4} yield
\begin{align*}
\left \|Q_\sigma ( H^V + i )^{-1} \right \| \, \le \, C \, \alpha^{\frac 12} \, \sigma^{\frac 12}.
\end{align*}
By \eqref{eq:a2}, we have shown that $H_\sigma^V$ converges to $H^V$, as $\sigma \searrow 0$, in the norm resolvent sense.
\end{proof}

%****************

\section{Proof of Theorem \ref{thm-main-1-0}}

In this section, we prove the bound in Theorem \ref{thm-main-1-0},
which compares the full dynamics to the effective dynamics for the system without
infrared cutoff.
We have that
\eqn
	\lefteqn{
	\| \, e^{-it H^V} \,\JopN^\rho (\,  u_0^\e \, )
	\, - \,\JopN^\rho (\,  e^{-itH_{{\rm eff}}} \, u_0^\e \, ) \, \|_{\cH}
	}
	\nonumber\\
	&\leq&
	\| \, e^{-it\He} \,\Jop^\rho (\,  u_0^\e \, )
	\, - \, \Jop^\rho (\,   e^{-it\Heff} \, u_0^\e \, ) \, \|_{\cH}
	\nonumber\\
	&&+\,
	\| \, e^{-it\He} \,\Jop^\rho (\,  u_0^\e \, ) \, - \,
	 e^{-itH^V} \,\JopN^\rho (\,  u_0^\e \, ) \, \|_{\cH}
	 \nonumber\\
	 &&+ \,
	 \| \,   \Jop^\rho (\,   e^{-it\Heff} \, u_0^\e \, )  \, - \,
	 \JopN^\rho (\,   e^{-it H_{\rm eff}} \, u_0^\e \, ) \, \|_{\cH} \,,
\eeqn
for any $t$ and $0 < \sigma < \rho \le 1$.
It follows from Theorem \ref{thm-main-1-1}  that the first term on the r.s. of the
inequality sign is bounded by $C_\delta \, (1+\ln(\rho^{-1})) \,  \e^{\frac23-\delta} \, t + C \, \alpha^{\frac12} \, \rho^{\frac12} \,   t \, (\e^\kappa +\e t ) $,
uniformly in $\sigma>0$.

From  Proposition \ref{prop-s-lim}, it follows that the second and third term
on the r.s. converge to zero, as $\sigma\searrow0$.
By taking $\sigma$ to zero, we thus conclude that
\eqn
	\| \, e^{-it H^V} \,\JopN^\rho (\,  u_0^\e \, )
	\, - \,\JopN^\rho (\,  e^{-itH_{{\rm eff}}} \, u_0^\e \, ) \, \|_{\cH}
	\, \leq \, C_\delta \, (1+\ln(\rho^{-1})) \,  \e^{\frac23-\delta} \, t +
	C \, \alpha^{\frac12} \, \rho^{\frac12} \,   t \, (\e^\kappa +\e t )  \, .
\eeqn
Due to our choice  $\rho = \e^{\frac23-\delta}$, this concludes the proof of Theorem \ref{thm-main-1-0}. We note that in the inequality \eqref{eq-main-thm-difference-0-1}, 
the logarithmic term $\ln(\rho_\e^{-1})$ has been 
absorbed by an arbitrary small shift of $\delta$, which we do not
keep track of notationally.
\qed

 $\;$ \\

\section{H\"older continuity of the ground state}
\label{sec-HolderContin-1}

We recall that $\Phi_\sigma^\rho(p)$ denotes a normalized ground state of the Bogoliubov transformed fiber Hamiltonian $K_\sigma^\rho(p) = W_{\nablE(p)}^\rho \, \Hn(p) \, (W_{\nablE(p)}^\rho)^*$, with infrared cutoff $\sigma>0$ (see \eqref{eq:Kn}). Our aim in this appendix is to prove that, for a suitable choice of the vectors $\Phi_\sigma^\rho(p)$, the map $p \mapsto \Phi_\sigma^\rho(p)$ is $\holder$-H\"older continuous, for $\holder < 2/3$.

For $\rho=1$, we set
\eqn\label{PhiK}
\Phi_\sigma(p) := \Phi_\sigma^1(p), \qquad K_\sigma(p) := K_\sigma^1(p).
\eeqn
We remark that
\eqn
K_\sigma^\rho(p) \, = \, \big ( W_{\nabla E_\sigma(p)}^{\rho,1} \big )^* \, K_\sigma(p) \, W_{\nabla E_\sigma(p)}^{\rho,1} \, , \qquad \Phi_\sigma^\rho(p) \,  = \, \big ( W_{\nabla E_\sigma(p)}^{\rho,1} \big )^* \,\Phi_\sigma(p) \, , \label{eq:dd3}
\eeqn
where we recall that $W_{\nablE(p)}^{\rho,1}$ is defined in \eqref{eq-Weylop-def-1}.

Letting
\eqn\label{Fs}
	\Fo_\sigma \, := \, \bigoplus_{n\geq0} \, \, {\rm Sym} 
	( \, L^2( \{ k \in \mathbb{R}^3 , |k| \ge \sigma \} \, \times \, \{+,-\} \, ) \, )^{\otimes n}
\eeqn
denote the Fock space of photons of energies $\ge \sigma$, and identifying $\Fo_\sigma$ with a subspace of $\Fo$, we observe that $K_\sigma(p)$ leaves $\Fo_\sigma$ invariant. Let $ \widetilde K_\sigma(p) $ denote the restriction of $K_\sigma(p)$ to $\Fo_\sigma$. An important property, proven in \cite{BFP,cfp2,frpi}, is that there is an energy gap of size $\eta \sigma$,  where  $\eta>0$  is uniform in $\sigma \searrow 0$ , in the spectrum of $ \widetilde K_\sigma(p) $ above the ground state energy $E_\sigma(p)$. Moreover, one can choose
\eqn
	\Phi_\sigma(p) \, = \,   \widetilde \Phi_\sigma(p) \otimes \Omega_{<\sigma}  \, , \label{eq:PhitildePhi}
\eeqn
in the representation $\Fo \simeq \Fo_\sigma \otimes \Fo_{<\sigma}$, where
\eqn
	\Fo_{<\sigma} \, := \, \bigoplus_{n\geq0} \, \, {\rm Sym} ( \, L^2( \{ k \in \mathbb{R}^3 , |k| \le \sigma \} \, \times \, \{+,-\} \, ) \, )^{\otimes n}  \, .
\eeqn

Now, let $\Omega_\sigma$ denote the vacuum sector in $\Fo_\sigma$ and $\widetilde \Pi_\sigma(p)$ be the rank-one projection onto the eigenspace associated with $E_\sigma(p) = \inf \mathrm{spec}( \tilde K_\sigma(p) )$. By \cite{cfp2,frpi}, 
\begin{align}\label{eq:|PiOmega|}
\| \widetilde \Pi_\sigma(p) \Omega_\sigma \| \, \ge \, \frac{1}{3},
\end{align}
for arbitrary $\sigma>0$ and $|p|\le1/3$ provided that $\alpha$ is chosen sufficiently small.
 Then $\widetilde \Phi_\sigma(p)$ can be chosen in the following way:
\begin{align}\label{eq:def_tildePhi}
\widetilde \Phi_\sigma(p) = \frac{ \widetilde \Pi_\sigma(p) \Omega_\sigma }{ \| \widetilde \Pi_\sigma(p) \Omega_\sigma \| }.
\end{align}

Let $N$ denote the number operator,
\eqn
	N \, = \, \sum_\lambda \int d k \, b^*_\lambda(k) \, b_\lambda(k) \, .
\eeqn
The following proposition has been proven in \cite{chfr,cfp2,frpi}.
\begin{proposition}\label{prop:convergence_GS}
For $\alpha \ll 1$ and $|p| \le 1/3$, there exists a  normalized  vector $\Phi(p)$ in the Fock space  $\Fo$  such that $\Phi_\sigma(p) \to \Phi(p)$, strongly, as $\sigma\to0$. The following bound holds,
\begin{align}
\| N^{\frac12} \Phi_\sigma(p) \| \, \le \, C \alpha^{\frac12} \, , \label{eq:boundN}
\end{align}
uniformly in $\sigma \ge 0$.
Moreover, For all $\delta>0$, there exists $\alpha_\delta>0$ and $C_\delta < \infty$ such that, for all $0 \le \alpha \le \alpha_\delta$, $0 \le \sigma' < \sigma \le 1$ and $|p| \le 1/3$,
\begin{align}
& \| \Phi_\sigma(p) - \Phi_{\sigma'}(p) \| \, \le \, C_\delta \, \alpha^{\frac{1}{4}} \, \sigma^{1-\delta}, \\
& | \nabla E_\sigma(p) - \nabla E_{\sigma'}(p) | \, \le \, C_\delta \, \alpha^{\frac{1}{4}} \, \sigma^{1-\delta}. \label{eq:bb1}
\end{align}
\end{proposition}
As a consequence, we show the following corollary.
\begin{corollary}\label{cor:convergence_GS}
Let $0 < \rho \le 1$. For all $\delta>0$, there exists $0<\alpha_\delta\ll1$ such that, for all $0 \le \alpha \le \alpha_\delta$ and $|p| \le 1/3$, there exists a vector $\Phi^\rho(p)$ in the Fock space such that $\Phi_\sigma^\rho(p) \to \Phi^\rho(p)$, strongly, as $\sigma\to0$. Moreover, there exists a constant $C_\delta < \infty$ such that, for all $0 \le \alpha \le \alpha_\delta$, $0 \le \sigma' < \sigma \le 1$ and $|p| \le 1/3$,
\begin{align}
& \| \Phi_\sigma^\rho(p) - \Phi_{\sigma'}^\rho(p) \| \, \le \, C_\delta \, \alpha^{\frac{1}{4}} \, \sigma^{1-\delta} \, (1 + \alpha^{\frac12} \ln(\rho^{-1}) ) \, .
\end{align}
\end{corollary}
\begin{proof}
Using \eqref{eq:dd3}, we split
\eqn \label{eq:aa1}
\Phi_\sigma^\rho(p) - \Phi_{\sigma'}^\rho(p) \, = \, \Big ( \big ( W_{\nabla E_\sigma(p)}^{\rho,1} \big )^* - \big ( W_{\nabla E_{\sigma'}(p)}^{\rho,1} \big )^* \Big )  \Phi_\sigma(p) \, + \, \big ( W_{\nabla E_{\sigma'}(p)}^{\rho,1} \big )^* \big ( \Phi_\sigma(p) - \Phi_{\sigma'}(p) \big ). \label{eq:aa2}
\eeqn
By Proposition \ref{prop:convergence_GS} and unitarity of $W_{\nabla E_\sigma(p)}^{\rho,1}$, the second term is estimated as
\eqn
\Big \| \big ( W_{\nabla E_{\sigma'}(p)}^{\rho,1} \big )^* \big ( \Phi_\sigma(p) - \Phi_{\sigma'}(p) \big ) \Big \| \le C_\delta \, \alpha^{\frac{1}{4}} \, \sigma^{1-\delta} \, . \label{eq:aa3}
\eeqn
The first term in the right side of \eqref{eq:aa1} is estimated as
\begin{align}
 \Big \|Ê\Big ( \big ( W_{\nabla E_\sigma(p)}^{\rho,1} \big )^* - \big ( W_{\nabla E_{\sigma'}(p)}^{\rho,1} \big )^* \Big )  \Phi_\sigma(p) \Big \| & = \Big \|Ê\Big ( \mathbf{1} - W_{\nabla E_\sigma(p)}^{\rho,1} \big ( W_{\nabla E_{\sigma'}(p)}^{\rho,1} \big )^* \Big )  \Phi_\sigma(p) \Big \| \notag \\
& \le \big \| B(\rho) \Phi_\sigma(p) \big \| \, , \label{eq:aa4}
\end{align}
by unitarity of $W_{\nabla E_\sigma(p)}^{\rho,1}$ and the spectral theorem, where
\begin{align}
B(\rho) := \alpha^{\frac12}\sum_\lambda\int_{\rho \le |k| \le 1}
	dk \,
	\Big ( \frac{\nablE(p)\cdot\e_\lambda(k) \, b_\lambda(k) -h.c.}{|k|^{1/2}(|k|-\nablE(p)\cdot k)} - \frac{\nabla E_{\sigma'}(p)\cdot\e_\lambda(k) \, b_\lambda(k) -h.c.}{|k|^{1/2}(|k|-\nabla E_{\sigma'}(p)\cdot k)} \Big ) \, . \label{eq:defB(rho)}
\end{align}
To estimate $\| B(\rho) \Phi_\sigma(p) \|$, we use the well known fact that, for any $f \in \mathrm{L}^2 ( \mathbb{R}^3 \times \{ + , - \} )$,
\eqn
	 \| a^\#(f) ( N + 1 )^{-\frac 12} \| \, \le \, \sqrt{2} \| f \|_{L^2} \, . \label{eq:cc1}
\eeqn
Clearly,
\begin{align}
& \frac{\nablE(p)\cdot\e_\lambda(k)}{|k|^{1/2}(|k|-\nablE(p)\cdot k)} - \frac{\nabla E_{\sigma'}(p)\cdot\e_\lambda(k) }{|k|^{1/2}(|k|-\nabla E_{\sigma'}(p)\cdot k)} \notag \\
& = \, \frac{( \nablE(p) - \nabla E_{\sigma'}(p) ) \cdot\e_\lambda(k)}{|k|^{1/2}(|k|-\nablE(p)\cdot k)} + \frac{\nabla E_{\sigma'}(p)\cdot\e_\lambda(k) }{|k|^{1/2}(|k|-\nabla E_{\sigma}(p)\cdot k)} \frac{ ( \nabla E_\sigma(p) - \nabla E_{\sigma'}(p) ) \cdot k }{(|k|-\nabla E_{\sigma'}(p)\cdot k)} \, .
\end{align}
Hence, by \eqref{eq:bb1} and the facts that $| \nabla E_\sigma(p) | , | \nabla E_{\sigma'}(p) | \le 1/2$ for $\alpha$ small enough (see Proposition \ref{prp-En-properties} (2)), we obtain
\begin{align}
\Big | \frac{\nablE(p)\cdot\e_\lambda(k)}{|k|^{1/2}(|k|-\nablE(p)\cdot k)} - \frac{\nabla E_{\sigma'}(p)\cdot\e_\lambda(k) }{|k|^{1/2}(|k|-\nabla E_{\sigma'}(p)\cdot k)} \Big | \, \le \,  \frac{  C_\delta \, \alpha^{\frac{1}{4}} \, \sigma^{1-\delta} }{ |k|^{\frac32} } \, .
\end{align}
Thus, \eqref{eq:defB(rho)} and \eqref{eq:cc1} yield that
\begin{align}
\big \| B(\rho) \Phi_\sigma(p) \big \| \, &\le \,   C_\delta \, \alpha^{\frac{3}{4}} \, \sigma^{1-\delta} \Big \| \frac{ \mathbf{1}_{ \rho \le |k| \le 1 }( |k| ) }{ |k|^{\frac32} } \Big \|_{ L^2_k } \big \| (N+1)^{\frac12} \Phi_\sigma(p) \big \| \, \notag \\
& \le \,   C_\delta \, \alpha^{\frac{3}{4}} \, \sigma^{1-\delta} \, \ln(\rho^{-1}) \, .
\end{align}
where we used \eqref{eq:boundN} in the last inequality. Together with \eqref{eq:aa2} -- \eqref{eq:aa4}, this concludes the proof of  Corollary~\ref{cor:convergence_GS}.
\end{proof}
The following result follows from \cite{cfp2,frpi} (it is also a consequence of \eqref{eq:boundd2E} in Proposition \ref{prp-En-properties} (2)).
\begin{proposition}\label{prop:C2_Esigma}
There exist $\alpha_c>0$ and $C>0$ such that, for all $0\le\alpha\le\alpha_c$ and $p,p'$ satisfying $|p| \le 1/3$, $|p'|\le1/3$,
\begin{align}
& \big | \nabla E_\sigma(p) - \nabla E_\sigma(p') \big | \, \le \, C \, | p - p' |,
\label{eq:nablaE(P)-nablaE(P')}
\end{align}
uniformly in $\sigma > 0$.
\end{proposition}
We now prove the following proposition.
\begin{proposition}\label{prop:Holder}
Let $0 < \rho \leq 1$. For all $\delta > 0$, there exist $\alpha_\delta>0$ and $C_\delta < \infty$ such that, for all $0\le\alpha\le\alpha_\delta$, $\sigma>0$ and $p,k \in \mathbb{R}^3$ satisfying $|p|\le1/3$, $|p+k|\le1/3$,
\begin{align}
\| \Phi_\sigma^\rho(p+k) - \Phi_\sigma^\rho(p) \| \, \le \, C_\delta \, (1+ \alpha^{\frac12} \ln(\rho^{-1})) \, |k|^{\frac{2}{3}-\delta} .
\end{align}
\end{proposition}
\noindent \textbf{Proof.} \\ 
\textit{Step 1}. ~ We first prove that, for all $0 < \sigma < \rho \le 1$,
\begin{align}
\| \Phi_\sigma^\rho(p+k) - \Phi_\sigma^\rho(p) \| \, \le \, C \, |k| \, ( \sigma^{-\frac{1}{2}} + \alpha^{\frac12}Ê\ln ( \rho^{-1} ) )  \, .  \label{eq:Holder1}
\end{align}
We decompose
\begin{align}
\Phi_\sigma^\rho(p+k) - \Phi_\sigma^\rho(p) \, &= \,  \big ( W_{\nabla E_\sigma(p+k)}^{\rho,1} \big )^* \,\Phi_\sigma(p+k) \, - \, \big ( W_{\nabla E_\sigma(p)}^{\rho,1} \big )^* \,\Phi_\sigma(p) \notag \\
& = \,  \Big ( \big ( W_{\nabla E_\sigma(p+k)}^{\rho,1} \big )^* - \big ( W_{\nabla E_\sigma(p)}^{\rho,1} \big )^* \Big ) \,\Phi_\sigma(p) \notag \\
&\quad + \, \big ( W_{\nabla E_\sigma(p+k)}^{\rho,1} \big )^* \, \big ( \Phi_\sigma(p+k) - \Phi_\sigma(p) \big ) \, . \label{eq:bb2}
\end{align}
To estimate the first term in the right side of \eqref{eq:bb2}, we proceed as in the proof of Corollary \ref{cor:convergence_GS}. Namely, we have that
\begin{align}
 \Big \|Ê\Big ( \big ( W_{\nabla E_\sigma(p+k)}^{\rho,1} \big )^* - \big ( W_{\nabla E_{\sigma}(p)}^{\rho,1} \big )^* \Big )  \Phi_\sigma(p) \Big \| & = \Big \|Ê\Big ( \mathbf{1} - W_{\nabla E_\sigma(p+k)}^{\rho,1} \big ( W_{\nabla E_{\sigma}(p)}^{\rho,1} \big )^* \Big )  \Phi_\sigma(p) \Big \| \notag \\
& \le \, \big \| C(\rho) \Phi_\sigma(p) \big \| \, , \label{eq:bb3}
\end{align}
by the spectral theorem, where
\begin{align}
C(\rho) := \alpha^{\frac12}\sum_\lambda\int_{\rho \le |\tilde k| \le 1}
	d \tilde k \,
	\Big ( \frac{\nabla E_\sigma(p+k)\cdot\e_\lambda(\tilde k) \, b_\lambda(\tilde k) -h.c.}{|\tilde k|^{1/2}(|\tilde k|-\nabla E_\sigma(p+k)\cdot \tilde k)} - \frac{\nabla E_{\sigma}(p)\cdot\e_\lambda(\tilde k) \, b_\lambda(\tilde k) -h.c.}{|\tilde k|^{1/2}(|\tilde k|-\nabla E_{\sigma}(p)\cdot \tilde k)} \Big ) \, . \notag
\end{align}
Using Proposition \ref{prop:C2_Esigma}, one verifies that
\begin{align}
\Big | \frac{\nabla E_\sigma(p+k)\cdot\e_\lambda(\tilde k)}{|\tilde k|^{1/2}(|\tilde k|-\nabla E_\sigma(p+k)\cdot \tilde k)} - \frac{\nabla E_{\sigma}(p)\cdot\e_\lambda(\tilde k) }{|\tilde k|^{1/2}(|\tilde k|-\nabla E_{\sigma}(p)\cdot \tilde k)} \Big | \, \le \,  \frac{ C \, |k| }{ |\tilde k|^{\frac32} } \, .
\end{align}
Hence \eqref{eq:cc1} implies that
\begin{align}
\big \| C(\rho) \Phi_\sigma(p) \big \| \, &\le \, C \, \alpha^{\frac12} \, |k| \, \Big \| \frac{ \mathbf{1}_{ \rho \le | \tilde k | \le 1 }( \tilde k ) }{ |\tilde k|^{\frac32} } \Big \|_{ L^2_{\tilde k} } \big \| (N+1)^{\frac12} \Phi_\sigma(p) \big \| \, \notag \\
& \le \, C \, \alpha^{\frac12} \, |k| \, \ln(\rho^{-1}) \, , \label{eq:bb4}
\end{align}
where we used \eqref{eq:boundN} in the last inequality. Equations \eqref{eq:bb3} and \eqref{eq:bb4} yield
\begin{align}
 \Big \|Ê\Big ( \big ( W_{\nabla E_\sigma(p+k)}^{\rho,1} \big )^* - \big ( W_{\nabla E_{\sigma}(p)}^{\rho,1} \big )^* \Big )  \Phi_\sigma(p) \Big \| & \le C \, \alpha^{\frac12}\,  |k| \, 
\ln(\rho^{-1}). 
\end{align}

It remains to estimate the second term in the right side of \eqref{eq:bb2}. By unitarity of $W_{\nabla E_\sigma(p+k)}^{\rho,1}$, it suffices to estimate $\| \Phi_\sigma(p+k) - \Phi_\sigma(p) \|$.
Using \eqref{eq:|PiOmega|} and the relation
 \begin{align}
\big \| (\widetilde \Pi_\sigma(p)  - \widetilde \Pi_\sigma(p+k) )\varphi \big \|^2&=\langle  \varphi , (\widetilde \Pi_\sigma(p+k) + \widetilde \Pi_\sigma(p) - \widetilde \Pi_\sigma(p) \widetilde \Pi_\sigma(p+k) - \widetilde \Pi_\sigma(p+k) \widetilde \Pi_\sigma(p) ) \varphi \rangle  \notag \\
&=\langle  \varphi , (\widetilde{\Pi}^\perp_\sigma(p+k)  \widetilde \Pi_\sigma(p) + \widetilde{\Pi}^\perp_\sigma(p) \widetilde \Pi_\sigma(p+k) ) \varphi \rangle , \notag \\
&=\langle  \varphi , ( \widetilde \Pi_\sigma(p) \widetilde{\Pi}^\perp_\sigma(p+k)  \widetilde \Pi_\sigma(p) + \widetilde{\Pi}^\perp_\sigma(p) \widetilde \Pi_\sigma(p+k) \widetilde{\Pi}^\perp_\sigma(p) ) \varphi \rangle , \notag \\
&= \| \widetilde{\Pi}^\perp_\sigma(p+k)  \widetilde \Pi_\sigma(p) \varphi \|^2 + \| \widetilde \Pi_\sigma(p+k) \widetilde{\Pi}^\perp_\sigma(p) \varphi \|^2 , \notag 
\end{align}
for any $\varphi \in \Fo_\sigma$, where $\widetilde{\Pi}^\perp_\sigma(p) := I - \widetilde \Pi_\sigma(p)$,  we obtain that
\begin{align}
\| \Phi_\sigma(p+k) - \Phi_\sigma(p) \| \, &= \, \| \widetilde \Phi_\sigma(p+k) - \widetilde \Phi_\sigma(p) \| \, \notag \\
&\le \, \frac{ 2 }{ \| \widetilde \Pi_\sigma(p) \Omega_\sigma \| }  \big \| (\widetilde \Pi_\sigma(p)  - \widetilde \Pi_\sigma(p+k) )\Omega_\sigma \big \| \notag \\
&\le \, 6 \big \| \widetilde \Pi_\sigma(p)  - \widetilde \Pi_\sigma(p+k) \big \| \notag \\
&\le \, 6 ( \big \| \widetilde{\Pi}^\perp_\sigma( p+k ) \widetilde \Pi_\sigma(p) \| + \| \widetilde{\Pi}^\perp_\sigma( p ) \widetilde \Pi_\sigma(p+k) \big \| ) \notag \\
& \le \, 6 ( \| \widetilde{\Pi}^\perp_\sigma(p+k) \widetilde \Phi_\sigma(p) \| + \| \widetilde{\Pi}^\perp_\sigma(p) \widetilde \Phi_\sigma(p+k) \| ). \label{eq:Holder2}
\end{align}

Since there is an energy gap of size
$\eta \sigma$ above $E_\sigma(p+k)$ in the spectrum of the operator
$\widetilde K_\sigma(p+k)$, we can estimate
$$
\widetilde \Pi_\sigma^\perp(p+k) \; \leq \; 
\frac{1}{\eta \sigma} \, \big( \widetilde K_\sigma(p+k) - E_\sigma(p+k)\big),
$$
and hence
\begin{align}
\| \widetilde{\Pi}^\perp_\sigma(p+k) \widetilde \Phi_\sigma(p) \|
\; \leq \; 
\frac{2}{\eta^{1/2} \sigma^{1/2}} \, 
\big\| (\widetilde K_\sigma(p+k) - E_\sigma(p+k)\big)^{1/2} 
\widetilde \Phi_\sigma(p) \big\| \, . \label{eq:new3}
\end{align}
We have by \eqref{eq:Kn}, \eqref{PhiK}, the definition after \eqref{Fs} and \eqref{Hp} 
\eqn
 \widetilde K_\sigma( p + k )= \widetilde K_\sigma( p  )  +	k \cdot \nabla_p  \widetilde K_\sigma(p)  \, + k^2 / 2,
 \eeqn
 where $\nabla_p \widetilde K_\sigma(p) \, := \, W_{\nablE(p)}^1 \, \nabla_p \Hn(p) \, (W_{\nablE(p)}^1)^*$, with $\nabla_p\Hn(p) \, := \, p \, - \, \Pf \, - \, \alpha^{\frac12}\Af$. 
Using this expansion and  the Feynman-Hellman formula,
\eqn
	\langle \tilde \Phi_\sigma(p) , \nabla_p  \widetilde K_\sigma(p)  \tilde \Phi_\sigma(p) \rangle \, = \, \nabla E_\sigma(p) \, ,
\eeqn
together with the mean-value theorem and Proposition \ref{prop:C2_Esigma}, we have that (see also \cite[Lemma 3.6]{cffs})
\begin{align}
& \big \| (  \widetilde K_\sigma( p + k )  - E_\sigma( p + k ) )^{\frac{1}{2}}  \widetilde \Phi_\sigma(p)  \big \|^2 \notag \\
& = \, \big \langle   \widetilde \Phi_\sigma(p)  , (  \widetilde K_\sigma( p + k )  - E_\sigma( p + k ) )   \widetilde \Phi_\sigma(p)  \big \rangle  \notag \\
& = \, \big \langle   \widetilde \Phi_\sigma(p)  , (  \widetilde K_\sigma( p )  + k \cdot ( \nabla_p  \widetilde K_\sigma( p )  ) + k^2/2 - E_\sigma( p + k ) )   \widetilde \Phi_\sigma(p)  \big \rangle \notag \\
& = \, E_\sigma( p ) - E_\sigma( p + k ) + k \cdot ( \nabla_p E_\sigma(p) ) + k^2 / 2 \notag \\
&  
= \, \tfrac{1}{2} k^2 + \int_0^1 k \cdot \big[ \nabla_p E_\sigma(p) - 
\nabla_p E_\sigma(p + \tau k) \big] \, d\tau
\notag \\
& \le \, C \, k^2. \label{eq:nablaK-nablaE_j_2}
\end{align}
Hence,
\begin{align}
\big \| \big (  \widetilde K_\sigma(p+k)  - E_\sigma(p+k) \big )^{\frac{1}{2}} \tilde \Phi_\sigma(p) \big \| \, \le C \, |k|. \label{eq:Holder4}
\end{align}
Combining \eqref{eq:new3} and \eqref{eq:Holder4}, we obtain that
\begin{align}
\| \widetilde{\Pi}^\perp_\sigma(p+k) \widetilde \Phi_\sigma(p) \| 
 \, \le \, C \, |k| \, \sigma^{-\frac12} \, .
\end{align}
Proceeding in the same way, it follows likewise that
\begin{align}
\| \widetilde{\Pi}^\perp_\sigma(p) \widetilde \Phi_\sigma(p+k) \| 
 \, \le \, C \, |k| \, \sigma^{-\frac12} \, ,
\end{align}
and hence, by \eqref{eq:Holder2}, \eqref{eq:Holder1} follows. \\

\noindent \textit{Step 2}. ~ We now prove that $\| \Phi_\sigma^\rho(p+k) - \Phi_\sigma^\rho(p) \| \, \le \, C_\delta \, (1+ \alpha^{\frac12} \ln(\rho^{-1})) \, |k|^{\frac{2}{3}-\delta}$ (with $C_\delta < \infty$ for $\delta > 0$).

Suppose first that $\sigma \ge |k|^{2/3}$. Then by Step 1, we have that
\begin{align}
\| \Phi_\sigma^\rho(p+k) - \Phi_\sigma^\rho(p) \| \, \le \, C \, |k| \, \big ( |k|^{-\frac{1}{3} } + \alpha^{\frac12} \, \ln( \rho^{-1} ) \big ) \, = \, C \, |k|^{ \frac{2}{3} } + C \, \alpha^{\frac12} \, \ln( \rho^{-1} ) \, |k|.
\end{align}
 Conversely,  assume that $\sigma \, \le \, |k|^{2/3}$. We write
\begin{align}
\| \Phi_\sigma^\rho(p+k) - \Phi_\sigma^\rho(p) \| \, &\le \, \| \Phi_\sigma^\rho(p+k) - \Phi^\rho(p+k) \|+ \| \Phi^\rho( p+k ) - \Phi_{|k|^{2/3}}^\rho(p+k) \| \notag \\
&\quad + \| \Phi_\sigma^\rho(p) - \Phi^\rho(p) \|+ \| \Phi^\rho( p ) - \Phi_{|k|^{2/3}}^\rho(p) \| \notag \\
&\quad + \| \Phi_{|k|^{2/3}}^\rho(p+k) - \Phi_{|k|^{2/3}}^\rho(p) \|.
\end{align}
By Corollary \ref{cor:convergence_GS}, the first two lines are bounded by
\begin{align}
 &\| \Phi_\sigma^\rho(p+k) - \Phi^\rho(p+k) \|+ \| \Phi^\rho( p+k ) - \Phi_{|k|^{2/3}}^\rho(p+k) \| \notag \\
& +  \| \Phi_\sigma^\rho(p) - \Phi^\rho(p) \|+ \| \Phi^\rho( p ) - \Phi_{|k|^{2/3}}^\rho(p) \| \notag \\
 & \le \, C_\delta \, \alpha^{\frac14} \, \big (1+ \alpha^{\frac12} \, \ln(\rho^{-1}) \big ) \, |k|^{\frac23(1-\delta)} \, ,
\end{align}
whereas by Step 1, the last term is bounded by $C \, |k|^{ \frac{2}{3} } + C \, \alpha^{\frac12} \, \ln( \rho^{-1} ) \, |k|$. Setting $\delta' = 2\delta/3$ and changing notations concludes the proof of the proposition. \qed
 $\;$ \\

%\appendix

%****************

\appendix

\section{Proof of Estimate \eqref{HK-est}}
\label{ssec-Dop2-2}
In this Appendix, we prove \eqref{HK-est}.
It asserts that
\eqn \label{HK-est2}
	\|(\Kn^\rho(p) \, - \, \Hn(p))\Phsig^\rho(p)\|_\Fo
	\, \leq \, C \, \alpha^{\frac12} \, \rho^{\frac12} \, |p| \, ,
\eeqn
for all $p\in\cS$, for a constant  $C < \infty$ independent of $\alpha$, $\sigma$, and $\rho$,  where $0 < \sigma < \rho \le 1$.

To begin with, let
\eqn
	v_\lambda^\sharp(k) \, := \,
	\alpha^{\frac12} \, \mathbf{1}_{\sigma \le |k| \le \rho }(|k|) \,
	\frac{\nablE(p)\cdot\e_\lambda^\sharp(k)}{|k|^{1/2}(|k|-\nablE(p)\cdot k)} \,,
\eeqn
(scalar-valued) and
\eqn
	w_\lambda^\sharp(k) \, := \,
	\alpha^{\frac12} \, \mathbf{1}_{\sigma \le |k| \le 1 }(|k|) \,
	\frac{ \e_\lambda^\sharp(k)}{|k|^{1/2} }
\eeqn
(vector-valued).
We note that
\eqn
	|v_\lambda(k)| \, \leq \, C \, \alpha^{\frac 12} \, |p| \, 
	\frac{ \mathbf{1}_{\sigma \le |k| \le \rho }(|k|) }{|k|^{\frac32}} 
\label{eq:f1}
\eeqn
and
\eqn
	|w_\lambda(k)| \, \leq C \, \alpha^{\frac 12} \, \frac{ \mathbf{1}_{\sigma \le |k| \le 1 }(|k|) }{|k|^{\frac12}}
\eeqn
where we have used that $|\nablE(p)| \, \leq \, C \, |p|$, uniformly in the infrared cutoff $0 \leq \sigma \leq 1$.

Using that
\eqn
	 W_{\nablE(p)}^{\s,\rho} \, b_\lambda^\sharp(k) \, ( W_{\nablE(p)}^{\s,\rho} )^* \, = \,
	 b_\lambda^\sharp(k) \, + \, v_\lambda^\sharp(k) \,,
\eeqn
a straightforward calculation yields
\eqn
	\lefteqn{
	\Kn^\rho(p) \, - \, \Hn(p)
	}
	\nonumber\\
	 & = & W_{\nablE(p)}^{\s,\rho} \, \Hn(p) \, ( W_{\nablE(p)}^{\s,\rho} )^* \, - \, \Hn(p)
	\nonumber\\
	&=&  2  V(p) \cdot (\nabla_p\Hn(p)) 
           \, + \, V^2(p) \, + \, Y(p) \,,
\eeqn
where 
\eqn
	\nabla_p\Hn(p) \, = \, p \, - \, \Pf \, - \, \alpha^{\frac12}\Af \,, \label{eq:defnablaH}
\eeqn
with
\eqn\label{eq:defAsigma}
	\Af \, = \, \sum_\lambda\big(b_\lambda(w_\lambda)+b_\lambda^*(w_\lambda)\big) \,,
\eeqn
and
\eqn
	V(p) &:=& \sum_\lambda\Big[ \, b_\lambda(k v_\lambda)
	\, + \, b_\lambda^*(k v_\lambda) \, + \, 2Re (w_\lambda,v_\lambda)
	\, + \, (v_\lambda, k v_\lambda) \, \Big] \, ,
\eeqn
(vector-valued operator) and
\eqn
	Y(p) &:=& \sum_\lambda\Big[ \, 
      b_\lambda\big(  (k^2 + |k|)  v_\lambda\big) 
      \, + \, 
      b_\lambda^*\big(  (k^2 + |k|)  v_\lambda\big) 
      \, + \, (v_\lambda, |k| v_\lambda) \, 
       \, + \, 2Re (k \cdot w_\lambda, v_\lambda)\Big]  \, ,
\eeqn
(scalar-valued operator).  Note 
that both $V(p)$ and $Y(p)$ are proportional to $|\nablE(p)|$
since all terms are of first or higher order in $v_\lambda$ (which is proportional to $|\nablE(p)| \le C|p|$).
 
Using Lemma~\ref{lm:b1} and \eqref{eq:f1}, we observe that
\eqn \label{eq-vp}
\| V(p) (H_f + 1)^{-1/2} \| 
& \leq &
2 \big\| (|k| + |k|^2)^{\frac12} v_\lambda \big\|_{L^2}
+ \big\| |k|^{\frac12} v_\lambda \big\|_{L^2}^2 
+ \big\| w_\lambda  \, v_\lambda \big\|_{L^1}
\nonumber \\
& \leq &
C \, \alpha^{1/2} \, |p| \, \rho^{1/2} \, ,  
\eeqn	
and similarly
\eqn \label{eq-vp2}
\| V(p)^2 (H_f + 1)^{-1} \| 
& \leq & 
C \, \alpha \, |p|^2 \, \rho \, ,  
\\ \label{eq-yp}
\| Y(p) (H_f + 1)^{-1/2} \| 
& \leq &
C \, \alpha^{1/2} \, |p| \, \rho \, .
\eeqn	
Next we note that for any normalized vector $\Phi \in D(H(p))$, 
we have the estimate
\eqn
\big\| (\tfrac{1}{2}(p-P_f)^2 + H_f + 1) \, \Phi \big\| 
& \leq &
\big\| (H_\sigma(p) + 1) \, \Phi \big\| 
+ \alpha^{1/2} \big\| A_\sigma \cdot \nabla H_\sigma(p) \, \Phi \big\| 
+ \alpha \big\| A_\sigma^2 \, \Phi \big\| 
\nonumber \\[1ex]
& \leq &
\big\| (H_\sigma(p) + 1) \, \Phi \big\| 
+ C \, \alpha^{1/2} 
\big\| (H_f + 1)^{1/2} \nabla H_\sigma(p) \, \Phi \big\| 
\nonumber \\
&  &
+ C \, \alpha \big\| (H_f + 1) \, \Phi \big\| \, .
\eeqn	
Since furthermore $P_f$ and $H_f$ commute, we have that 
$(H_f + 1)^2 \leq (\tfrac{1}{2}(p-P_f)^2 + H_f + 1)^2$ and hence
\eqn \label{eq-hf-hp}
\big\| (H_f + 1) \, \Phi \big\| 
& \leq &
2 \big\| (H_\sigma(p) + 1) \, \Phi \big\| 
+ C \, \alpha^{1/2} 
\big\| (H_f + 1)^{1/2} \nabla H_\sigma(p) \, \Phi \big\| \, ,
\eeqn	
provided $\alpha >0$ is sufficiently small.
Now, we observe that
\eqn
\big\| [ H_f \, , \, \nabla H_\sigma(p) ] \, (H_f+1)^{-1/2} \big\|
& = &
\alpha^{1/2} \big\| [ H_f \, , \, A_\sigma ] \, (H_f+1)^{-1/2} \big\|
\ \leq \ 
C \, \alpha^{1/2} ,
\eeqn	
which implies that
\eqn
\| (H_f + 1)^{1/2} \nabla H_\sigma(p) \, \Phi \|^2
& = &
\big \langle  \nabla H_\sigma(p) \Phi \cdot \, , \: 
(H_f + 1) \, \nabla H_\sigma(p) \, \Phi \big\rangle
\\ \nonumber 
& = &
\big\langle  \nabla H_\sigma(p)^2 \Phi \, , \: 
(H_f + 1) \, \Phi \big\rangle
- 
\big\langle  [ H_f \, , \, \nabla H_\sigma(p) ] \, \Phi \cdot \, , \: 
\nabla H_\sigma(p) \, \Phi \big\rangle
\\ \nonumber 
& \leq &
C \big\|  H_\sigma(p) \, \Phi \big  \| \, \: 
\big \| (H_f + 1) \, \Phi \big\|
+ C \, \alpha^{1/2} \| (H_f + 1)^{1/2} \nabla H_\sigma(p) \, \Phi \| \, .
\eeqn	
Hence, for sufficiently small $\alpha >0$, we have that
\eqn \label{eq-hf-hp2,5}
\| (H_f + 1)^{1/2} \nabla H_\sigma(p) \, \Phi \|
& \leq &
C \, \| H_\sigma(p) \, \Phi \|^{1/2} \, \| (H_f + 1) \, \Phi \|^{1/2} \, .
\eeqn	
Inserting this estimate into \eqref{eq-hf-hp}, we obtain
for all normalized $\Phi$ that
\eqn \label{eq-hf-hp2}
\big\| (H_f + 1) \, \Phi \big\| 
& \leq &
C \big\| (H_\sigma(p) + 1) \, \Phi \big\| \, ,
\eeqn	
and, additionally using \eqref{eq-hf-hp2,5}, that
\eqn \label{eq-hf-hp3}
\| (H_f + 1)^{1/2} \nabla H_\sigma(p) \, \Phi \|
& \leq &
C \big\| (H_\sigma(p) + 1) \, \Phi \big\| \, ,
\eeqn	
provided $\alpha >0$ is sufficiently small.

We arrive at the assertion by applying Estimates~\eqref{eq-vp},
\eqref{eq-vp2}, \eqref{eq-yp}, \eqref{eq-hf-hp2}, and \eqref{eq-hf-hp3},
\eqn \label{eq-est-1}
\big\| \big( \Kn^\rho(p) - \Hn(p) \big) \, \Phi_\sigma^\rho(p) \big\|
& \leq & 
2 \| V(p) \cdot \nabla_p\Hn(p) \, \Phi_\sigma^\rho(p) \|
+ \| V(p)^2 \, \Phi_\sigma^\rho(p) \|
+ \| Y(p) \, \Phi_\sigma^\rho(p) \|
\nonumber \\[1ex]
& \leq & 
2 \| V(p) (H_f + 1)^{-1/2} \| \; 
\| (H_f + 1)^{1/2} \nabla_p\Hn(p) \, \Phi_\sigma^\rho(p) \|
\nonumber \\
&  & 
\; + \; 
\| V(p)^2  (H_f + 1)^{-1} \| \; 
\| (H_f + 1) \, \Phi_\sigma^\rho(p) \|
\nonumber \\
&  & 
\; + \; 
\| Y(p) (H_f + 1)^{-1/2} \| \; 
\| (H_f + 1) \, \Phi_\sigma^\rho(p) \|
\nonumber \\[1ex]
& \leq & 
C \, \alpha^{1/2} \, |p| \, \rho^{1/2} \,
\big\| (H_\sigma(p) + 1) \, \Phi_\sigma^\rho(p) \big\| \, 
\nonumber \\[1ex]
& \leq & 
C' \, \alpha^{1/2} \, |p| \, \rho^{1/2} \, ,
\eeqn	
which is Inequality~\eqref{HK-est2} or \eqref{HK-est}, respectively.
\qed
\\

\parskip = 0 pt
\parindent = 0 pt

\end{document}